\documentclass[11pt]{article}

\usepackage{algorithmic}
\usepackage{algorithm}
\usepackage{amsfonts}
\usepackage{amsmath} 
\usepackage{latexsym} 
\usepackage{epsfig}
\usepackage{latexsym} 
\usepackage{verbatim}
\begin{document}

\newtheorem{theorem}{Theorem}
\newtheorem{corollary}[theorem]{Corollary}
\newtheorem{prop}[theorem]{Proposition} 
\newtheorem{problem}[theorem]{Problem}
\newtheorem{lemma}[theorem]{Lemma} 
\newtheorem{remark}[theorem]{Remark}
\newtheorem{observation}[theorem]{Observation}
\newtheorem{defin}{Definition} 
\newtheorem{example}{Example}
\newtheorem{conj}{Conjecture} 
\newenvironment{proof}{{\bf Proof:}}{\hfill$\Box$} 
\newcommand{\PR}{\noindent {\bf Proof:\ }} 
\def\EPR{\hfill $\Box$\linebreak\vskip.5mm} 
 
\def\Pol{{\sf Pol}} 
\def\mPol{{\sf MPol}} 
\def\Polo{{\sf Pol}_1} 
\def\PPol{{\sf pPol\;}} 
\def\Inv{{\sf Inv}}
\def\mInv{{\sf MInv}} 
\def\Clo{{\sf Clo}\;} 
\def\Con{{\sf Con}} 
\def\concom{{\sf Concom}\;} 
\def\End{{\sf End}\;}
\def\Sub{{\sf Sub}\;} 
\def\Im{{\sf Im}} 
\def\Ker{{\sf Ker}\;} 
\def\H{{\sf H}}
\def\S{{\sf S}} 
\def\D{{\sf P}} 
\def\I{{\sf I}} 
\def\Var{{\sf var}} 
\def\PVar{{\sf pvar}} 
\def\fin#1{{#1}_{\rm fin}}
\def\P{{\sf P}} 
\def\Pfin{{\sf P_{\rm fin}}} 
\def\Id{{\sf Id}}
\def\R{{\rm R}} 
\def\F{{\rm F}} 
\def\Term{{\sf Term}}
\def\var#1{{\sf var}(#1)} 
\def\Sg#1{{\sf Sg}(#1)} 
\def\Sgo#1{{\sf Sg}_{\mathrm{old}}(#1)} 
\def\Sgn#1{{\sf Sg}_{\mathrm{new}}(#1)} 
\def\Sgg#1#2{{\sf Sg}_{#1}(#2)} 
\def\Cg#1{{\sf Cg}(#1)}
\def\tol{{\sf tol}}
\def\lnk{{\sf lk}}
\def\rbcomp#1{{\sf rbcomp}(#1)}
  
\let\cd=\cdot 
\let\eq=\equiv 
\let\op=\oplus 
\let\omn=\ominus
\let\meet=\wedge 
\let\join=\vee 
\let\tm=\times
\def\ldiv{\mathbin{\backslash}} 
\def\rdiv{\mathbin/}
\def\lnk{{\sf lk}}

\def\typ{{\sf typ}} 
\def\zz{{\un 0}} 
\def\zo{{\un 1}}
\def\one{{\bf1}} 
\def\two{{\bf2}} 
\def\three{{\bf3}}
\def\four{{\bf4}} 
\def\five{{\bf5}}
\def\pq#1{(\zz_{#1},\mu_{#1})}
  
\let\wh=\widehat 
\def\ox{\ov x} 
\def\oy{\ov y} 
\def\oz{\ov z}
\def\of{\ov f} 
\def\oa{\ov a} 
\def\ob{\ov b} 
\def\oc{\ov c}
\def\od{\ov d} 
\def\oob{\ov{\ov b}} 
\def\rx{{\rm x}}
\def\rf{{\rm f}} 
\def\rrm{{\rm m}} 
\let\un=\underline
\let\ov=\overline 
\let\cc=\circ 
\let\rb=\diamond 
\def\ta{{\tilde a}} 
\def\tz{{\tilde z}}
  
  
\def\zZ{{\mathbb Z}} 
\def\B{{\mathcal B}} 
\def\P{{\mathcal P}}
\def\zL{{\mathbb L}} 
\def\zD{{\mathbb D}}
 \def\zE{{\mathbb E}}
\def\zG{{\mathbb G}} 
\def\zA{{\mathbb A}} 
\def\zB{{\mathbb B}}
\def\zC{{\mathbb C}} 
\def\zM{{\mathbb M}} 
\def\zR{{\mathbb R}}
\def\zS{{\mathbb S}} 
\def\zT{{\mathbb T}} 
\def\zN{{\mathbb N}}
\def\zQ{{\mathbb Q}} 
\def\zW{{\mathbb W}} 
\def\bK{{\bf K}}
\def\C{{\bf C}} 
\def\M{{\bf M}} 
\def\E{{\bf E}} 
\def\N{{\bf N}}
\def\O{{\bf O}} 
\def\bN{{\bf N}} 
\def\bX{{\bf X}} 
\def\GF{{\rm GF}} 
\def\cC{{\mathcal C}} 
\def\cA{{\mathcal A}}
\def\cB{{\mathcal B}} 
\def\cD{{\mathcal D}} 
\def\cE{{\mathcal E}} 
\def\cF{{\mathcal F}} 
\def\cG{{\mathcal G}} 
\def\cH{{\mathcal H}}
\def\cI{{\mathcal I}} 
\def\cK{{\mathcal K}} 
\def\cL{{\mathcal L}} 
\def\cP{{\mathcal
P}} \def\cR{{\mathcal R}} 
\def\cRY{{\mathcal RY}}
\def\cS{{\mathcal S}} 
\def\cT{{\mathcal T}} 
\def\oB{{\ov B}}
\def\oC{{\ov C}} 
\def\ooB{{\ov{\ov B}}} 
\def\ozB{{\ov{\zB}}}
\def\ozD{{\ov{\zD}}} 
\def\ozG{{\ov{\zG}}}
\def\tcA{{\widetilde\cA}} 
\def\tcC{{\widetilde\cC}}
\def\tcF{{\widetilde\cF}} 
\def\tcI{{\widetilde\cI}}
\def\tB{{\widetilde B}} 
\def\tC{{\widetilde C}}
\def\tD{{\widetilde D}} 
\def\ttB{{\widetilde{\widetilde B}}}
\def\ttC{{\widetilde{\widetilde C}}}
\def\tba{{\tilde\ba}} 
\def\ttba{{\tilde{\tilde\ba}}}
\def\tbb{{\tilde\bb}} 
\def\ttbb{{\tilde{\tilde\bb}}}
\def\tbc{{\tilde\bc}} 
\def\tbd{{\tilde\bd}}
\def\tbe{{\tilde\be}} 
\def\tbt{{\tilde\bt}}
\def\tbu{{\tilde\bu}} 
\def\tbv{{\tilde\bv}}
\def\tbw{{\tilde\bw}} 
\def\tdl{{\tilde\dl}} 
\def\ocP{{\ov\cP}}
\def\tzA{{\widetilde\zA}} 
\def\tzC{{\widetilde\zC}}
\def\new{{\mbox{\footnotesize new}}}
\def\old{{\mbox{\footnotesize old}}}
\def\prev{{\mbox{\footnotesize prev}}}
\def\oo{{\mbox{\sf\footnotesize o}}}
\def\pp{{\mbox{\sf\footnotesize p}}}
\def\nn{{\mbox{\sf\footnotesize n}}} 
\def\oR{{\ov R}}
\def\bA{\mathbf{A}}
\def\bB{\mathbf{B}}
  
  
\def\gA{{\mathfrak A}} 
\def\gV{{\mathfrak V}} 
\def\gS{{\mathfrak S}} 
\def\gK{{\mathfrak K}} 
\def\gH{{\mathfrak H}}
  
\def\ba{{\bf a}} 
\def\bb{{\bf b}} 
\def\bc{{\bf c}} 
\def\bd{{\bf d}} 
\def\be{{\bf e}} 
\def\bbf{{\bf f}} 
\def\bg{{\bf g}}
\def\bh{{\bf h}}
\def\bi{{\bf i}} 
\def\bm{{\bf m}} 
\def\bo{{\bf o}} 
\def\bp{{\bf p}} 
\def\bs{{\bf s}} 
\def\bu{{\bf u}} 
\def\bt{{\bf t}} 
\def\bv{{\bf v}} 
\def\bx{{\bf x}}
\def\by{{\bf y}} 
\def\bw{{\bf w}} 
\def\bz{{\bf z}}
\def\ga{{\mathfrak a}} 
\def\oal{{\ov\al}} 
\def\obeta{{\ov\beta}}
\def\ogm{{\ov\gm}} 
\def\oep{{\ov\varepsilon}}
\def\oeta{{\ov\eta}} 
\def\oth{{\ov\th}} 
\def\ovm{{\ov\mu}}
\def\ozero{{\ov0}}
  
  
\def\CCSP{\hbox{\rm c-CSP}} 
\def\CSP{{\rm CSP}} 
\def\NCSP{{\rm \#CSP}} 
\def\mCSP{{\rm MCSP}} 
\def\FP{{\rm FP}} 
\def\PTIME{{\bf PTIME}} 
\def\GS{\hbox{($*$)}} 
\def\ry{\hbox{\rm r+y}}
\def\rb{\hbox{\rm r+b}} 
\def\Gr#1{{\mathrm{Gr}(#1)}}
\def\Grp#1{{\mathrm{Gr'}(#1)}} 
\def\Grpr#1{{\mathrm{Gr''}(#1)}}
\def\Scc#1{{\mathrm{Scc}(#1)}} 
\def\rel{R} 
\def\relo{Q}
\def\rela{S} 
\def\dep{\mathsf{dep}}
\def\Filt{\mathrm{Ft}}
\def\Filts{\mathrm{Fts}} 
\def\Agr{$\mathbb{A}$}
\def\Al{\mathrm{Alg}}
\def\Sig{\mathrm{Sig}}
\def\strat{\mathsf{strat}}
\def\relmax{\mathsf{relmax}}
\def\srelmax{\mathsf{srelmax}}
\def\Meet{\mathsf{Meet}}
\def\amax{\mathsf{amax}}
\def\max{\mathsf{max}}
\def\as{\mathsf{as}}
\def\sss{\mathsf{s}}
\def\star{\hbox{$(*)$}}
\def\bmal{{\mathbf m}}
\def\Af{\mathsf{Af}}
\let\sqq=\sqsubseteq
\def\se#1{\mathsf{s}(#1)}
\def\see#1#2{\mathsf{s}_{#1}(#2)}
\def\umax{\mathsf{umax}}
\def\as{\mathsf{as}}
\def\asm{\mathsf{asm}}

  
\let\sse=\subseteq 
\def\ang#1{\langle #1 \rangle}
\def\angg#1{\left\langle #1 \right\rangle}
\def\dang#1{\ang{\ang{#1}}} 
\def\vc#1#2{#1 _1\zd #1 _{#2}}
\def\tms#1#2{#1 _1\tm\dots\tm #1 _{#2}}
\def\zd{,\ldots,} 
\let\bks=\backslash 
\def\red#1{\vrule height7pt depth3pt width.4pt
\lower3pt\hbox{$\scriptstyle #1$}}
\def\fac#1{/\lower2pt\hbox{$\scriptstyle #1$}}
\def\me{\stackrel{\mu}{\eq}} 
\def\nme{\stackrel{\mu}{\not\eq}}
\def\eqc#1{\stackrel{#1}{\eq}} 
\def\cl#1#2{\arraycolsep0pt
\left(\begin{array}{c} #1\\ #2 \end{array}\right)}
\def\cll#1#2#3{\arraycolsep0pt \left(\begin{array}{c} #1\\ #2\\
#3 \end{array}\right)} 
\def\clll#1#2#3#4{\arraycolsep0pt
\left(\begin{array}{c} #1\\ #2\\ #3\\ #4 \end{array}\right)}
\def\cllll#1#2#3#4#5#6{ \left(\begin{array}{c} #1\\ #2\\ #3\\
#4\\ #5\\ #6 \end{array}\right)} 
\def\pr{{\rm pr}}
\let\upr=\uparrow 
\def\ua#1{\hskip-1.7mm\uparrow^{#1}}
\def\sua#1{\hskip-0.2mm\scriptsize\uparrow^{#1}} 
\def\lcm{{\rm lcm}} 
\def\perm#1#2#3{\left(\begin{array}{ccc} 1&2&3\\ #1&#2&#3
\end{array}\right)} 
\def\w{$\wedge$} 
\let\ex=\exists
\def\NS{{\sc (No-G-Set)}} 
\def\lev{{\sf lev}}
\let\rle=\sqsubseteq 
\def\ryle{\le_{ry}} 
\def\ryprec{\le_{ry}}
\def\os{\mbox{[}} 
\def\zs{\mbox{]}}
\def\link{{\sf link}}
\def\solv{\stackrel{s}{\sim}} 
\def\mal{\mathbf{m}}
\def\precs{\prec_{as}}

  
\def\lb{$\linebreak$}  
  
\def\ar{\hbox{ar}} 
\def\Im{{\sf Im}\;} 
\def\deg{{\sf deg}}
\def\id{{\rm id}}
  
\let\al=\alpha 
\let\gm=\gamma 
\let\dl=\delta 
\let\ve=\varepsilon
\let\ld=\lambda 
\let\om=\omega 
\let\vf=\varphi 
\let\vr=\varrho
\let\th=\theta 
\let\sg=\sigma 
\let\Gm=\Gamma 
\let\Dl=\Delta
  
  
\font\tengoth=eufm10 scaled 1200 
\font\sixgoth=eufm6
\def\goth{\fam12} 
\textfont12=\tengoth 
\scriptfont12=\sixgoth
\scriptscriptfont12=\sixgoth 
\font\tenbur=msbm10
\font\eightbur=msbm8 
\def\bur{\fam13} 
\textfont11=\tenbur
\scriptfont11=\eightbur 
\scriptscriptfont11=\eightbur
\font\twelvebur=msbm10 scaled 1200 
\textfont13=\twelvebur
\scriptfont13=\tenbur 
\scriptscriptfont13=\eightbur
\mathchardef\nat="0B4E 
\mathchardef\eps="0D3F

\title{Graphs of relational structures: restricted types}
\author{Andrei A.\ Bulatov}
\date{}
\maketitle
\begin{abstract}
The algebraic approach to the Constraint Satisfaction Problem (CSP) uses 
high order symmetries of relational structures --- polymorphisms --- to study
the complexity of the CSP. In this paper we further develop one of the 
methods the
algebraic approach can be implemented, and apply it to some kinds of the CSP.
This method was introduced in our LICS 2004 paper and involves the study of 
the local structure of finite algebras and relational structures. It associates 
with an algebra $\zA$ or a relational structure $\bA$ a graph, whose 
vertices are the elements of $\zA$ (or $\bA$), the edges represent subsets 
of $\zA$ such that the restriction of some term operation of $\zA$ is `good' 
on the subset, that is, act as an operation of one of the 3 types: semilattice, 
majority, or affine. In this paper we use this theory and consider algebras 
with edges from a restricted set of types. We prove type restrictions 
are preserved under the standard algebraic constructions. Then we show 
that if the types edges in a relational structure are restricted, then the 
corresponding CSP can be solved in polynomial time by specific algorithms. 
In particular, we give a new, somewhat more intuitive proof of the Bounded 
Width Theorem: the CSP over algebra $\zA$ has bounded width if and only 
if $\zA$ does not contain affine edges. Actually, this result shows that bounded 
width implies width (2,3). Finally, we prove that algebras without semilattice 
edges have few subalgebras of powers, that is, the CSP over such algebras 
is also polynomial time. The methods and results obtained in this paper are
important ingredients of the 2017 proof of the Dichotomy Conjecture by the 
author. The Dichotomy Conjecture was also proved independently by Zhuk.
\end{abstract}



\section{Introduction}\label{sec:introduction}

The Constraint Satisfaction Problem (CSP) has received a great deal of 
attention over the last several decades from various areas including logic, 
artificial intelligence, computer science, discrete mathematics, and algebra. 
Different facets of the CSP play an important role in all these disciplines. 
In this paper we focus on the complexity of and algorithms for the CSP. 
This direction in the study of the CSP revolves around the Dichotomy 
Conjecture by Feder and Vardi \cite{Feder93:monotone,Feder98:monotone} for the decision 
version of the problem, and the Unique Games Conjecture by Khot 
\cite{Khot02:power} for the optimization version. 

One of the several possible forms of the CSP asks whether there exists 
a homomorphism between two given relational structures. The Dichotomy 
Conjecture deals with the so called \emph{nonuniform} CSP parametrized 
by the target structure $\bB$; that is, for a given relational structure $\bA$ 
the goal is to decide the existence of a homomorphism from 
$\bA$ to the fixed target structure $\bB$. Such a problem is usually 
denoted by $\CSP(\bB)$. The conjecture claims  that every problem 
$\CSP(\bB)$ is either NP-complete, or is solvable in polynomial time; so no 
intermediate complexity class is attained by problems $\CSP(\bB)$. This 
conjecture has been attacked using different approaches, see, e.g.\ 
\cite{Kolaitis00:game,Kolaitis03:csp,Hell90:h-coloring,Hell04:homomorphism}, 
however, the algebraic approach using invariance properties of relational 
structures has been the most effective and eventually led to a resolution of the 
Dichotomy Conjecture. This approach is 
based on exploiting the properties of \emph{polymorphisms} of relational
structures, which can be thought of as homomorphisms from a power 
$\bA^n$ of a structure $\bA$ to the structure itself, but are usually 
viewed as multi-ary operations on $\bA$ `preserving' the relations of 
$\bA$. The use of polymorphisms was first proposed by Jeavons et al.\
\cite{Jeavons97:closure,Jeavons98:algebraic,Jeavons98:consist}, 
who showed that the complexity of $\CSP(\bB)$ is completely determined 
by the polymorphisms of $\bB$, and identified several types of 
polymorphisms whose presence guarantees the solvability of $\CSP(\bB)$ 
in polynomial time. We will be using these types of operations all the 
time in this paper, so we name them here: semilattice, majority, and 
affine operations, for exact definitions see Section~\ref{sec:edges}. The 
algebraic approach was later developed further in 
\cite{Bulatov05:classifying,Bulatov03:multi} to use universal algebras
associated with relational structures rather than polymorphisms; which 
allowed for applications of structural results from universal algebra. 
This connection has been used first to state the Dichotomy Conjecture 
in a precise form, that basically boils down to the presence of `nontrivial' 
polymorphisms (in which case $\CSP(\bB)$ is polynomial time solvable) 
\cite{Bulatov05:classifying}, and to obtain a number of strong tractability
and dichotomy results \cite{Bulatov02:maltsev-3-element,%
Bulatov06:3-element,Bulatov06:simple,Bulatov11:conservative,%
Bulatov16:conservative,Barto11:conservative,%
Barto12:near,Barto14:local,Idziak10:few}. This line of research recently 
culminated in confirming the Dichotomy Conjecture 
\cite{Bulatov17:dichotomy,Zhuk17:proof}.

One of the main obstacles that had to be overcome to prove the Dichotomy 
Conjecture is that structural theories of universal algebras existed before
are not designed for the CSP. Therefore, the study of the CSP has triggered 
substantial research in algebra aiming to obtain more advanced results 
on the structure of finite algebras. Several approaches have been suggested. 
The first one is based on the absorbing properties of algebras, 
see, e.g.\ \cite{Barto12:absorbing,Barto15:constraint}. Within this 
approach the bounded width conjecture has been proved \cite{Barto14:local} 
(see more about this conjecture in subsequent sections), along with many 
algebraic results and generalizations of the known CSP complexity results 
\cite{Barto11:conservative,Barto12:near,Barto14:local,Barto18:finitely}. 
Another potential approach is to use so called \emph{key relations},
i.e.\ relations that cannot be represented through a combination of 
simpler ones, see, e.g.\ \cite{Zhuk14:key}; although this method requires 
further development. The third approach has been originally introduced in 
\cite{Bulatov04:graph,Bulatov08:recent,Bulatov11:conjecture} and uses 
the local structure of universal algebras. More precisely, it identifies 
small sets of elements of a relational structure or an algebra --- in most 
cases 2-element sets --- such that there is a polymorphism of the 
structure or a term operation of the algebra that behaves well on this 
subset, where `well' means that the operation is close to a semilattice, 
majority, or an affine one. These subsets are then considered edges of 
a graph; these edges can have one of the three types, corresponding 
to the three types of good operations: semilattice, majority, or affine. 
For a relational structure $\bA$ or an algebra $\zA$ the resulting graph 
will be denoted by $\cG(\bA)$ and $\cG(\zA)$, respectively. Properties of
graph $\cG(\zA)$ reflect many aspects of the corresponding CSP. For example, 
that for every algebra $\zA$ that gives rise to a tractable CSP, the graph 
$\cG(\zA)$ is connected, moreover, the types of edges present in the graph 
are related to other properties of the CSP. In particular, the absence of 
affine edges corresponds to the bounded width of the CSP. 

Although this paper does not directly deal with the Dichotomy Conjecture,
the methods and results obtained here are essential ingredients in the 
proof in \cite{Bulatov17:dichotomy}. We refine and advance the approach from 
\cite{Bulatov04:graph,Bulatov08:recent}. The main motif of this work is 
to consider algebras $\zA$ for which the graph $\cG(\zA)$ contains edges 
from a restricted set of types. We first show that the property to have 
edges from a certain set of types is preserved under the standard algebraic 
constructions.

\begin{theorem}\label{the:restrict-variety-int}
Let $S\sse\{\text{semilattice},\text{majority},\text{affine}\}$ and $\zA$ 
be a finite idempotent algebra such that every edge of $\cG(\zA)$ has a 
type from $S$. Then every edge of any finite algebra from the variety 
generated by $\zA$ belongs to $S$.
\end{theorem}

An algebra $\zA$ is said to have \emph{few subpowers} if the number of 
subalgebras of direct products of several copies of $\zA$ is exponentially 
smaller than it generally can be. Idziak et al.\ \cite{Idziak10:few} proved 
that if $\zA$ 
has few subpowers then $\CSP(\zA)$ can be solved in polynomial time. 
Moreover, for such CSPs it is possible to construct a small (polynomial size) 
set of generators of the set of all solutions to the problem. We show that 
every algebra whose edges are majority or affine has few subpowers, 
although it is not true that every algebra with few subpowers satisfies this 
condition.

\begin{theorem}\label{the:few-subpowers-int}
Let $\zA$ be an idempotent algebra every edge of which is majority or affine.
Then $\zA$ has few subpowers. In particular, $\CSP(\zA)$ can be solved in 
polynomial time.
\end{theorem}

Then we study algebras whose edges are either of the affine and semilattice 
types, or of the majority and semilattice types. In both case we show that 
$\cG(\zA)$ has stronger connectivity properties. In the semilattice-majority 
case we also give a somewhat more intuitive proof for the characterization of 
CSPs of bounded width than that 
in \cite{Barto14:local} and \cite{Bulatov09:bounded}.

\begin{theorem}\label{the:bounded-width-int}
Let $\zA$ be an idempotent algebra every edge of which is semilattice or 
majority. Then $\CSP(\zA)$ has bounded width. Moreover, every algebra 
that gives rise to a CSP of bounded width satisfies this condition.
\end{theorem}

\section{Preliminaries}\label{sec:preliminaries}

\subsection{Relational structures and universal algebras}\label{sec:csp}

By $[n]$ we denote the set $\{1\zd n\}$. For sets $\vc An$ tuples from 
$\tms An$ are denoted in boldface, say, $\ba$; the $i$th component of 
$\ba$ is referred to as $\ba[i]$. An $n$-ary relation $\rel$ over sets 
$\vc An$ is any subset of $\tms An$. For $I=\{\vc ik\}\sse[n]$ by 
$\pr_I\ba,\pr_I\rel$ we denote the \emph{projections} 
$\pr_I\ba=(\ba[i_1]\zd\ba[i_k])$, $\pr_I\rel=\{\pr_I\ba\mid\ba\in\rel\}$ 
of tuple $\ba$ and relation $\rel$. If $\pr_i\rel=A_i$ for each $i\in[n]$, 
relation $\rel$ is said to be a \emph{subdirect product} of $\tms An$. 
As usual, a \emph{relational structure} $\bA$ with a \emph{(relational) 
alphabet} $(\vc \rel m)$ is a set $A$ equipped with interpretations 
$\rel_i^\bA$ of predicate symbols with relations over $A$ of matching 
arity.

We assume familiarity with basic concepts of universal algebra, for 
references see \cite{Burris81:universal}. A (universal) algebra $\zA$ 
with a \emph{functional alphabet} $\vc f\ell$ is a set $A$, called the 
\emph{universe} equipped with interpretations $f_i^\zA$ of functional 
symbols with (multi-ary) operations on $A$ of matching arity. 
In this paper all structures and algebras are assumed finite. 
Algebras with the same functional alphabet are said to be \emph{similar}. 
Operations that can be derived from $\vc f\ell$ by means of composition 
are called \emph{term operations}. 

Let $\zA,\zB$ be similar algebras with 
universes $A$ and $B$, respectively. A mapping $\vf:A\to B$ is a 
\emph{homomorphism} of algebras, if it preserves all the operations, 
that is $\vf(f^\zA_i(\vc ak))=f_i^\zB(\vf(a_1)\zd\vf(a_k))$ for any 
$i\in[\ell]$ and  any $\vc ak\in A$.  A  bijective homomorphism is an 
\emph{isomorphism}. A set $B\sse A$ is a \emph{subuniverse} of 
$\zA$ if, for every $i\in[\ell]$, the operation $f_i^\zA$ restricted to $B$ 
takes values from $B$ only. For a nonempty subuniverse $B$ of algebra 
$\zA$ the algebra $\zB$ with universe $B$ and operations 
$f_1^\zB\zd f_\ell^\zB$ (where $f_i^\zB$ is a restriction of 
$f_i^\zA$ to $B$) is a \emph{subalgebra} of $\zA$. Given similar algebras 
$\zA,\zB$, a \emph{product} $\zA\tm \zB$ of $\zA$ and $\zB$ is the algebra 
similar to $\zA$ and $\zB$ with universe $A\tm B$ and operations computed 
coordinate-wise.  An algebra $\zC$ is a \emph{subdirect product} of $\zA$ 
and $\zB$ if it is a subalgebra of $\zA\tm\zB$ whose universe is a subdirect 
product of $A$ and $B$. An equivalence relation $\th$ on $A$ is called a 
\emph{congruence} of algebra $\zA$ if $\th$ is a subalgebra of
$\zA\tm\zA$. Given  a  congruence $\th$ on $\zA$   we  can  form  the  
\emph{factor  algebra} $\zA\fac\th$ similar to $\zA$,  whose  elements  
are  the equivalence  classes  of $\th$ and  the  operations  are  defined  
so  that  the natural projection mapping is a homomorphism 
$\zA\to\zA\fac\th$. The $\th$-block containing element $a\in\zA$ is denoted 
by $a\fac\th$. We often abuse the notation and use the same operation 
symbol for all similar algebras including factor algebras. In particular, 
to make notation lighter we use $f$ rather than $f\fac\th$ for operations
on a factor algebra. Algebra $\zA$ is \emph{simple}, if it has the trivial 
congruences only (i.e.\  the equality relation and the full congruence).  
If $\th$ is a maximal congruence of $\zA$, then $\zA\fac\th$ is simple. 

A \emph{variety} is a class of algebras closed under direct products 
(including infinite products), subalgebras, and homomorphic images 
(or factor algebras). Algebra $\zA$ is said to be idempotent if 
$f_i(x\zd x)=x$ for all $x\in A$ and any $i\in[\ell]$. If $\th$ is a congruence 
of an idempotent algebra $\zA$, then $\th$-blocks are subuniverses of 
$\zA$. The subalgebra of $\zA$ generated by a set $B\sse\zA$ is denoted 
$\Sgg\zA B$. In most cases $\zA$ is clear from the context and is omitted.

The connection between algebras and relational structures is given by the 
invariance relation. Let $\vc An$ be sets, operation $f(\vc xk)$ is defined on each 
of the $A_i$, and $\rel$ is a relation over $\vc An$. An operation $f(\vc xk)$ 
is said to \emph{preserve} relation $\rel$, or $f$ is a \emph{polymorphism} 
of $\rel$, or $\rel$ is \emph{invariant} with respect to $f$, if for any 
$\vc\ba k\in\rel$ the tuple $f(\vc\ba k)\in\rel$. Operation $f$ on a set $A$ 
is a polymorphism of  relational structure $\bA=(A;\vc\rel m)$ if it is a 
polymorphism of every relation of $\bA$. This definition can be generalized 
to multi-sorted relational structures, but we do not need it here. For a (finite) 
class of finite algebras $\cK$ with basic operations $\vc f\ell$ 
by $\Inv(\cK)$ we denote the class of all finitary relations over the universes 
of algebras from $\cK$ invariant under every $f_i$, $i\in[\ell]$. Alternatively, 
$\Inv(\cK)$ is the class of subalgebras of direct products of algebras 
from $\cK$.

\subsection{Constraint Satisfaction Problem}

The (\emph{nonuniform}) \emph{Constraint Satisfaction Problem} 
(\emph{CSP}) associated with a relational structure $\bB$ is the problem 
$\CSP(\bB)$, in which, given a structure $\bA$ of the same alphabet 
as $\bB$, the goal is to decide whether or not there is a homomorphism
from $\bA$ to $\bB$. Nonuniform CSPs can also be defined for algebras. 
For a class of algebras $\cK=\{\zA_i\mid i\in I\}$ 
for some set $I$ an instance of  $\CSP(\cK)$ is a triple $(V,\dl,\cC)$, 
where $V$ is a set of variables; $\dl:V\to\cK$ is a type function that 
associates every variable with a domain in $\cK$. Finally, $\cC$ is a 
set of constraints, i.e.\ pairs $\ang{\bs,\rel}$, where $\bs=(\vc vk)$ is 
a tuple of variables from $V$, and $\rel\in\Inv(\cK)$, a subset of 
$A_{\dl(v_1)}\tm\dots A_{\dl(v_k)}$. The goal is to find a solution, that is 
a mapping $\vf:V\to\bigcup\cK$ such that $\vf(v)\in\zA_{\dl(v)}$ and for every 
constraint $\ang{\bs,\rel}$, $\vf(\bs)\in\rel$. It is easy to see that if 
$\cK$ is a class containing just one algebra $\zA$, then $\CSP(\cK)$ 
can be viewed as the union of $\CSP(\bA)$ for all relational structures 
$\bA$ invariant under the operations of $\zA$.

The CSP dichotomy theorem 
\cite{Feder93:monotone,Bulatov17:dichotomy,Zhuk17:proof} states that for every 
relational structure $\bB$, $\CSP(\bB)$ is either solvable in polynomial time or 
is NP-complete. In its algebraic form \cite{Bulatov05:classifying} it claims that 
for any finite algebra $\zA$ the problem $\CSP(\zA)$ is either solvable in 
polynomial time or NP-complete; the single algebra $\zA$ can also be 
replaced here with a finite class of finite similar algebras. The algebraic 
approach also helps to make the tractability condition more precise: for a class 
$\cK$ of idempotent algebras the problem $\CSP(\cK)$ is solvable
in polynomial time if and only if the variety generated by $\cK$ does not 
contain `trivial' algebras, or, equivalently, when it omits type \one\ in the 
sense of tame congruence theory \cite{Hobby88:structure}. Otherwise 
$\CSP(\cK)$ is NP-complete. Note that all CSPs for non-idempotent 
algebras or relational structures are equivalent to some CSPs over 
idempotent algebras under log-space reductions 
\cite{Bulatov05:classifying}. In the next section we give an alternative 
characterization of algebras omitting type \one\ that will be used in this 
paper. In particular, all algebras we deal with will be assumed finite, 
idempotent, and omitting type \one.

\section{Coloured graphs}

\subsection{Edges}\label{sec:edges}
In \cite{Bulatov04:graph,Bulatov08:recent} we introduced a local approach 
to the structure of finite algebras. As we use this approach throughout the 
paper, we present it here in some details, see also 
\cite{Bulatov16:connectivity,Bulatov20:connectivity,Bulatov20:maximal}. 

Let $\zA$ be an algebra with universe $A$. Recall that a binary operation 
$f$ on $A$ is said to be \emph{semilattice} if it satisfies the equations 
$f(x,x)=x$, $f(x,y)=f(y,x)$, and $f(x,f(y,z))=f(f(x,y),z)$ for any $x,y,z\in A$. 
A ternary operation $g$ is said to be \emph{majority} if it satisfies the 
equations $g(x,x,y)=g(x,y,x)=g(y,x,x)=x$ for all $x,y\in A$. It is called 
\emph{Mal'tsev} if it satisfies $g(x,y,y)=g(y,y,x)=x$. An operation
is said to be semilatiice (majority, Mal'tsev) on a set $B\sse A$ or 
$B\sse A\fac\th$ for an equivalence relation $\th$, if the above equalities 
hold for all $x,y,z\in B$ or $x\fac\th,y\fac\th,z\fac\th\in B\fac\th$. 
A standard example of a Mal'tsev operation is the operation
$x-y+z$ of a module; we call this operation of a module \emph{affine}. 
Modules are used for definitions below, and so we need the following 
observation. Modules are not idempotent, and so in this paper they are 
replaced with their
\emph{full idempotent reducts}, in which we remove all the 
non-idempotent operations from the module. 

Graph $\cG(\zA)$ is introduced as follows. The vertex set
is the set $A$. A pair $ab$ of vertices is an \emph{edge} iff
there exists a congruence $\th$ of $\Sg{a,b}$, other than the full 
congruence and a term operation $f$ of $\zA$ such 
that either $\Sg{a,b}\fac\th$ is a set (that is, an algebra whose term 
operations are trivial), or $\Sg{a,b}\fac\th$ is a module and $f$ is an affine 
operation on it, or $f$ is a semilattice operation on
$\{a\fac\th,b\fac\th\}$, or $f$ is a majority operation on
$\{a\fac\th,b\fac\th\}$. (Note that we use the same operation symbol in this case.)

If there is a maximal congruence $\th$ such that $\Sg{a,b}\fac\th$ is 
a set then $ab$ is said to have the \emph{unary} type. If there are a 
maximal congruence $\th$ and a term operation of $\zA$ such that $f$ 
is a semilattice operation on $\{a\fac\th,b\fac\th\}$ then $ab$ is said 
to have the
{\em semilattice type}. An edge $ab$ is of {\em majority type} if there are 
a maximal congruence $\th$ and a term operation $f$
such that $f$ is a majority operation on $\{a\fac\th,b\fac\th\}$ and there is no 
semilattice term operation on $\{a\fac\th,b\fac\th\}$. Finally, $ab$ has the 
{\em affine type} if there are maximal $\th$ and $f$ such that $f$ is an affine 
operation on $\Sg{a,b}\fac\th$ and $\Sg{a,b}\fac\th$ is a module; in 
particular it implies that there is no semilattice or majority operation on 
$\{a\fac\th,b\fac\th\}$.  In all cases we say that congruence $\th$ 
\emph{witnesses} the type of edge $ab$. Observe that a pair $ab$ can 
still be an edge of more than one type as witnessed by different congruences.
We will often refer to the set $\{a\fac\th,b\fac\th\}$ as a \emph{thick}
edge.

The conditions for $\CSP(\zA)$ to be tractable --- `omitting type \one', --- and
the condition for $\CSP(\zA)$ to have bounded width (see Section~\ref{sec:sm}
for more details) --- `omitting types \one\ and \two ' --- can be characterized as 
follows: An idempotent algebra $\zA$ omits the type \one\ (the types 
\one\ and \two) if and only if $\cG(\zB)$ contains no edges of the unary type 
($\cG(\zB)$ does not contain edges of the unary and affine types) for every 
subalgebra $\zB$ of $\zA$, see Theorem~12 from \cite{Bulatov16:connectivity}.
In this paper we always assume that algebras do not contain edges of the 
unary type.

For the sake of the Dichotomy Theorem, it suffices to consider 
\emph{reducts} of an algebra $\zA$ omitting type \one, that is, algebras 
with same universe but reduced set of term operations, as long as the reducts 
also omit type \one. In particular, we are interested in reducts of $\zA$, in 
which semilattice and majority edges are subalgebras.

\begin{theorem}[Theorem~12, \cite{Bulatov20:connectivity}]%
\label{the:adding}
Let $\zA$ be an idempotent algebra. There exists a reduct $\zA'$ of
$\zA$ such that\\[2mm]
(1) if 
$\cG(\zA)$ does not contain edges of the unary type, then $\cG(\zA')$ 
does not contain edges of the unary type;\\[2mm]
(2) if $\cG(\zA)$ contains no edges of the unary and affine types, 
then $\cG(\zA')$ contains no edges of the unary and affine types.
\end{theorem}

An algebra $\zA$ such that $a\fac\th\cup b\fac\th$ is a subuniverse of $\zA$ for 
every semilattice or majority edge $ab$ of $\zA$ is called \emph{smooth}.

Operations witnessing the type of edges can be significantly uniformized.
As is proved in \cite{Bulatov04:graph,Bulatov16:connectivity,Bulatov20:connectivity}, 
for any
finite class $\cK$ of smooth algebras there are term operations $f,g,h$ of 
$\cK$ such that for every edge $ab$ of any $\zA\in\cK$ witnessed by a 
maximal congruence $\th$, $f$ is a 
semilattice operation on $\{a\fac\th,b\fac\th\}$ whenever $ab$ is a semilattice 
edge, $g$ is a majority operation on $\{a\fac\th,b\fac\th\}$ if $ab$ is
a majority edge, and $h$ is an affine operation operation on 
$\{a\fac\th,b\fac\th\}$ if $ab$ is an affine edge.

Unlike majority and affine operations, for a semilattice edge $ab$ and a 
congruence $\th$ of $\Sg{a,b}$ witnessing that, there can be semilattice 
operations acting differently on $\{a\fac\th,b\fac\th\}$,which corresponds to 
the two possible orientations of $ab$. In every such case by fixing 
operation $f$ introduced above we effectively 
choose one of the two orientations. In this paper we do not really care 
about what orientation is preferable. 

\subsection{Thin edges}
Edges as defined above are not always most effective. In \cite{Bulatov20:connectivity,Bulatov20:maximal} 
we therefore refine these notions. For the definitions given in this section
we first need to fix a finite class $\cK$ of similar smooth algebras that
do not contain edges of the unary type. To streamline the arguments we will 
assume that $\cK$ is closed under taking subalgebras and factor-algebras.
A pair $ab$ of elements of algebra $\zA\in\cK$ is 
called a \emph{thin semilattice edge} if $ab$ is a semilattice edge, and the 
congruence witnessing that is the equality relation. In other words, $f(a,a)=a$ 
and $f(a,b)=f(b,a)=f(b,b)=b$. We denote the fact that $ab$ is a thin 
semilattice edge by $a\le b$. Operation $f$ can be selected to have an 
additional property.

\begin{prop}[Proposition~24, \cite{Bulatov20:connectivity}]%
\label{pro:good-operation}
Let $\cK$ be as specified above. There is a binary
term operation $f$ of $\cK$ such that $f$ is a semilattice operation on
$\{a\fac\th,b\fac\th\}$ for every semilattice edge $ab$ of any 
$\zA\in\cK$, where 
congruence $\th$ witnesses that, and, for any $a,b\in\zA$, either 
$a=f(a,b)$ or the pair $(a,f(a,b))$ is a thin semilattice edge of $\zA$. 
Operation $f$ with this property will be denoted by a dot (think 
multiplication).
\end{prop}

Defining thin majority and affine edges requires a bit more work. 
A ternary term operation $g'$ of $\cK$ is said to satisfy the \emph{majority
condition} if it satisfies the identity  $g'(x,g'(x,y,y),g'(x,y,y))=g'(x,y,y)$ 
and $g'$ is a majority operation on every thick majority
edge of every algebra from $\cK$. A ternary term operation $h'$ is said 
to satisfy the \emph{minority
condition} if it satisfies the identity $h'(h'(x,y,y),y,y)=h'(x,y,y)$ and 
$h'$ is a Mal'tsev operation on every thick minority
edge of every algebra from~$\cK$.

A pair $ab$, $a,b\in\zA\in\cK$ is called a \emph{thin majority edge} if 
\begin{itemize}
\item[(*)] 
for any term operation $g'$ satisfying the majority condition
the subalgebras $\Sg{a,g'(a,b,b)},\Sg{a,g'(b,a,b)},\Sg{a,g'(b,b,a)}$
contain $b$.
\end{itemize}

Fix an operation $h$ satisfying the minority condition, in particular it 
satisfies the equation $h(h(x,y,y),y,y)=h(x,y,y)$. A pair $ab$, $a,b\in\zA\in\cK$, 
is called a \emph{thin affine edge} (with respect to $\cK$) if $h(b,a,a)=b$ and 
for every term operation $h'$ satisfying the minority condition
\begin{itemize}
\item[(**)] 
$b\in\Sg{a,h'(a,a,b)}$.
\end{itemize}
The operations $g,h$ introduced in the previous section, although they satisfy 
the majority and minority
conditions, respectively, do not 
have to satisfy any specific conditions on the set $\{a,b\}$, when
$ab$ is a thin majority or affine edge, except what follows from their 
definition. Also, both thin majority and thin affine edges are 
directed, since $a,b$ in the definition occur asymmetrically. We 
therefore can define yet another directed graph, 
$\cG'(\zA)$, in which the arcs are the thin edges of all types.

\begin{lemma}[Corollaries~25,29,33, \cite{Bulatov20:connectivity}]%
\label{lem:thin-semilattice}
Let $\zA$ be a smooth algebra.
Let $ab$ be a semilattice (majority, affine) edge, $\th$ a congruence of
$\Sg{a,b}$ that witnesses this. Then, for any $c\in a\fac\th$ there is 
$d\in b\fac\th$ such that $cd$ is a thin edge of the same type as $ab$.
\end{lemma}

We conclude this section with a result that shows that the presence of 
types of  thick and thin edges are closely related.

\begin{prop}[Proposition~34, \cite{Bulatov20:connectivity}]%
\label{pro:thin-thick-colors}
Let $\cK$ be a finite class of smooth algebras. Then there exists an 
algebra from $\cK$ containing a 
thin semilattice (majority, affine) edge if and only if there is an
algebra in $\cK$ containing a thick edge of the same type.
\end{prop}

\subsection{Paths and connectivity}
Let $\zA$ be a smooth algebra. A \emph{path} in $\zA$ is a sequence 
$a_0,a_1\zd a_k$ such 
that $a_{i-1}a_i$ is a thin edge for all $i\in[k]$ (note that thin edges 
are always assumed to be directed).
We will distinguish paths of several types depending on what types of
edges are allowed. If $a_{i-1}\le a_i$ for all $i\in[k]$ then the path is
called a \emph{semilattice} or \emph{s-path}. If for every $i\in[k]$
either $a_{i-1}\le a_i$ or $a_{i-1}a_i$ is a thin
affine edge then the path is called \emph{affine-semilattice} or
\emph{as-path}. Similarly, if only semilattice and thin majority edges
are allowed we have a \emph{semilattice-majority} or 
\emph{sm-path}. The path is called \emph{asm-path} when all types
of edges are allowed. If 
there is a path $a=a_0,a_1\zd a_k=b$ which is arbitrary (semilattice, 
affine-semilattice, semilattice-majority) then $a$ is said to be
\emph{asm-connected} (or \emph{s-connected}, or 
\emph{as-connected}, or \emph{sm-connected}) to $b$. We will also 
say that $a$ is \emph{connected} to $b$ if it is asm-connected. 
We denote this by $a\sqq^{asm}b$
(for asm-connectivity), $a\sqq b$, $a\sqq^{as}b$ and 
$a\sqq^{sm}b$ for s-, as-, and sm-connectivity, respectively. 

Let $\cG_s(\zA),\cG_{as}(\zA),\cG_{asm}(\zA)$ denote the digraph 
whose nodes are the elements of $\zA$, and the arcs are the thin 
semilattice edges (thin semilattice and affine edges, all thin edges, 
respectively). The strongly connected component of $\cG_s(\zA)$ 
containing $a\in\zA$ will be denoted by $\se a$. The set of strongly 
connected components of $\cG_s(\zA)$ are ordered in the natural 
way (if $a\le b$ then $\se a\le \se b$), the elements belonging to 
maximal ones will be called \emph{maximal}, and the set of all 
maximal elements from $\zA$ will be denoted by $\max(\zA)$. 

The strongly connected component of $\cG_{as}(\zA)$ containing 
$a\in\zA$ will be denoted by $\as(a)$. A maximal strongly connected 
component of this graph is called an
\emph{as-component}, an element from an as-component
is called \emph{as-maximal}, and the set of all as-maximal elements is
denoted by $\amax(\zA)$. 

Alternatively, maximal and as-maximal elements 
can be characterized as follows: an element $a\in\zA$ is 
maximal (as-maximal) if for every $b\in\zA$ such 
that $a\sqq b$ ($a\sqq^{as}b$)
it also holds that $b\sqq a$ ($b\sqq^{as}a$). 
Sometimes it will be necessary to specify what the algebra is, 
in which we consider maximal components or as-components, 
and the corresponding connectivity. 
In such cases we we will specify it by writing $\see \zA a$, 
$\as_\zA(a)$. For connectivity we will use
$a\sqq_\zA b$ and $a\sqq_\zA^{as}b$.

\begin{prop}[Corollary~11, Theorem~23, \cite{Bulatov20:maximal}]%
\label{pro:as-connectivity}
Let $\zA$ be an algebra omitting type \one. Then\\[1mm]
(1) any $a,b\in\zA$ are connected in $\cG_{asm}(\zA)$ with an oriented 
path;\\[1mm]
(2) any $a,b\in\max(\zA)$ (or $a,b\in\amax(\zA)$) are connected in 
$\cG_{asm}(\zA)$ with a directed path.
\end{prop}

The graphs $\cG(\zA_),\cG_s(\zA), \cG_{as}(\zA), \cG_{asm}(\zA)$ retain 
substantial amount of crucial 
information required for solving CSPs. They witness that the omitting type 
\one\ condition and the bounded width condition hold, and also can certify 
some other useful properties. However, in general they also erase much 
information about the algebra. As an extreme example, if $\zA$ is a prime 
algebra, that is, every possible operation on its universe is a term operation 
of $\zA$, then it satisfies all the CSP related conditions on an algebra.
It has few subpowers (see the Section~\ref{sec:am}), $\CSP(\zA)$ has 
bounded width, etc. But according to the definitions graphs $\cG(\zA)$ and 
$\cG_{asm}(\zA)$ have only semilattice edges that are oriented in an 
arbitrary way. 
In particular, there is no way to know from these graphs that $\zA$ has few 
subpowers, unless one picks a very special orientation of semilattice edges.


\begin{lemma}[Corollary~18, \cite{Bulatov20:maximal}]%
\label{lem:path-extension}
Let $\rel$ be a subdirect product of $\tms\zA n$ and $I\sse[n]$.
For any $\ba\in\rel$, and an s- (as-, asm-) path 
$\vc\bb k\in\pr_I\rel$ with $\pr_I\ba=\bb_1$, there is an
s- (as-, asm-) path $\vc{\bb'}\ell\in\rel$ such that 
$\bb'_1=\ba$ and $\pr_I\bb'_\ell=\bb_\ell$. 
\end{lemma}

We will usually apply Lemma~\ref{lem:path-extension} as follows.

\begin{corollary}\label{cor:path-extension}
Let $\rel$ be a subdirect product of smooth algebras $\zA_1,\zA_2$ and let
$C_1,C_2$ be as-components of $\zA_1,\zA_2$, respectively. Then either 
$\rel\cap(C_1\tm C_2)=\eps$ or $\rel\cap(C_1\tm C_2)$ is a subdirect 
product of $C_1\tm C_2$.
\end{corollary}

\subsection{Rectangularity}\label{sec:rectangularity}

Let $\rel\le\zA_1\tm\dots\tm\zA_k$ be a relation. Also, let 
$\tol_i(\rel)$ (or simply $\tol_i$ if $\rel$ is clear from the context),
$i\in[k]$, denote the \emph{link} tolerance
\begin{align*}
& \{(a_i,a'_i)\in\zA_i^2\mid (a_1\zd a_{i-1},a_i,a_{i+1}\zd a_k),\\
& \quad
(a_1\zd a_{i-1},a'_i,a_{i+1}\zd a_k)\in\rel, \text{ for some
$(a_1\zd a_{i-1},a_{i+1}\zd a_k)$}\}.
\end{align*}
Recall that a tolerance is said to be \emph{connected} if its transitive 
closure is the full relation. The transitive closure 
$\lnk_i(\rel)$ of $\tol_i(\rel)$, $i\in[k]$, is called the 
\emph{link congruence}, and it is, indeed, a congruence. A binary 
relation $\rel$ is said to be \emph{linked} if both $\lnk_1(\rel)$ and 
$\lnk_2(\rel)$ are total congruences.

In \cite{Bulatov20:maximal} we proved some `rectangularity' 
properties of relations with respect to as-components and link congruences.


\begin{prop}[Corollary~27, \cite{Bulatov20:maximal}]%
\label{pro:max-gen} 
Let $\rel$ be a subdirect product of $\zA_1$ and $\zA_2$, 
$\lnk_1(\rel),\lnk_2(\rel)$ the link congruences, and let $B_1,B_2$ 
be as-components of a $\lnk_1(\rel)$-block and a 
$\lnk_2(\rel)$-block, respectively, such that 
$\rel\cap(B_1\tm B_2)\ne\eps$. Then $B_1\tm B_2\sse\rel$.
\end{prop}

\section{Algebras with graphs of restricted types}\label{sec:restricted}

We start with showing that finite algebras in the variety generated by an
algebra $\zA$ can only contain edges of the types already in $\zA$. 
                                                                                                                               
Let $T\sse\{\text{semilattice, majority, affine}\}$. An algebra $\zA$ is said 
to be \emph{$T$-restricted} if every edge of $\zA$ has a type from $T$.
Note that if we set $\cK$ to be the class of all factor algebras of subalgebras
of $\zA$ then by Proposition~\ref{pro:thin-thick-colors} it makes no difference
whether we restrict the set of types of thick or thin edges.

\begin{theorem}\label{the:set-preservation}
Let $T\sse\{\text{semilattice, majority, affine}\}$ and $\cK$ a finite 
collection of similar smooth $T$-restricted algebras containing no edges 
of the unary type. Then every finite algebra from the variety generated 
by $\cK$ is $T$-restricted.
\end{theorem}

\begin{proof}
Every subalgebra of a $T$-restricted algebra is $T$-restricted, as it follows 
from the definition of types of edges. Let $\zA=\zA_1\tm\dots\tm\zA_n$ 
where all $\vc\zA n$ are $T$-restricted. Suppose there is an edge $\ba\bb$ 
in $\zA$ of type $z\in\{\text{semilattice, majority, affine}\}-T$, and $\th$ 
is the congruence of $\zB=\Sg{\ba,\bb}$ witnessing that. Let 
$I(\ba',\bb')=\{i\in[n]\mid \ba'[i]=\bb'[i]\}$ for $\ba',\bb'\in\zA$. As is easily 
seen, for any $\ba'\in\ba\fac\th, \bb'\in\bb\fac\th$, the congruence 
$\th'=\th\cap\zD^2$ of $\zD=\Sg{\ba',\bb'}$ witnesses that $\ba'\bb'$
is an edge of $\zA$ of the same type as $\ba\bb$. Therefore, tuples 
$\ba,\bb$ can be assumed to be such that $I=I(\ba,\bb)$
is maximal among pairs $\ba',\bb'$ with $\ba'\in\ba\fac\th$, $\bb'\in\bb\fac\th$. 
Then for any $\bc\in\zB$ and any $i\in I$, $\bc[i]=\ba[i]$. 

Take $i\in[n]-I$ and set $A'=\{\ba'[i]\mid \ba'\in\ba\fac\th\}$ and 
$B'=\{\bb'[i]\mid \bb'\in\bb\fac\th\}$. By the choice of $\ba,\bb$, 
$A'\cap B'=\eps$. We argue that this means that the projection $\eta$ of $\th$
on the $i$th coordinate, that is, the congruence of $\zC=\pr_i\zB$ given by 
the transitive closure of 
$$
\tol_i = \{(a,b)\mid \text{ for some $\bc,\bd\in\zB$,
$a=\bc[i]$, $b=\bd[i], (\bc,\bd)\in\th$}\}
$$
is nontrivial. Indeed, if $\ba\bb$ is a semilattice or majority edge, then 
$\th$ has only two congruence blocks, whose restrictions on $\zC$ are
$A',B'$, which are disjoint. If $\ba\bb$ is affine, then suppose there exist 
$\bc,\bd\in\zB$ such that $(\bc,\bd)\not\in\th$ but $\bc[i]=\bd[i]$. Then
we replace $\ba,\bb$ with $\bc,\bd$: $\bc\bd$ is an affine edge and 
$I(\ba,\bb)\subset I(\bc,\bd)$.
Every term operation that is semilattice, majority or
affine on $\zB\fac\th$ is semilattice, majority, or affine on $\zC\fac\eta$,
as well. Since $\zC$ is generated by $\ba[i],\bb[i]$, this pair is an edge of 
$\zA_i$ of type $z$, a contradiction.

Now suppose that $\zA$ is $T$-restricted and $\zB=\zA\fac\al$ for some 
congruence $\al$. Let $ab$, $a,b\in\zB$, be an edge of type 
$z\in\{\text{semilattice,}\lb \text{majority, affine}\}$ and $\th$ a maximal 
congruence of $\zC=\Sgg{\zB}{a,b}$ witnessing that. We will find 
$a',b'\in\zA$ such that $a'b'$ is an edge of $\zA$ of type $z$, see also 
Lemma~16 from \cite{Bulatov20:maximal}. Let
$\zC'=\bigcup_{c\in\zC}c$ (elements of $\zC$ are subsets of $\zA$), 
$\al'=\al\cap\zC'^2$,
and $\th'=\al'\join\th$, a congruence of $\zC'$. Choose $a',b'\in\zC'$ such that
$a'\in a$ and $b'\in b$. Let $\th''$ be the restriction of $\th'$ on 
$\zC''=\Sgg{\zA}{a',b'}$. Then clearly, $\zC''\fac{\th''}$
is isomorphic to $\zC\fac\th$, and therefore $\th''$ witnesses that $a'b'$ is an 
edge of type $z$ in $\zA$.
\end{proof}

\section{Affine and majority: few subpowers}\label{sec:am}

\subsection{Few subpowers}\label{sec:few-subpowers}

We call algebras without semilattice edges \emph{semilattice free}.
In this section we prove two results that relate semilattice free algebras to algebras 
with the property to have few subpowers. The 
few subpowers property has been introduced in \cite{Berman10:varieties}. 
Let $\zA$ be a finite algebra. Then $s_\zA(n)$ denotes the logarithm (base 2) 
of the number of subalgebras of $\zA^n$; and $g_\zA(n)$ is the least 
number $k$ such that for every subalgebra $\zB$ of $\zA^n$, $\zB$ has a 
generating set containing at most $k$ elements. Algebra $\zA$ is said to have 
\emph{few subpowers} if $s_\zA(n)$ is bounded by a single exponential
function in $n$, that is by $2^{O(p(n))}$, where $p$ is a polynomial in $n$. 
Note that the number of all subsets of $\zA^n$ is $|\zA|^{|\zA|^n}$.

Having few subpowers can be characterized by the presence of an 
\emph{edge} term \cite{Berman10:varieties}. A term operation $f$ in $k+1$ 
variables is called an edge term if the following $k$ identities are satisfied:
\begin{eqnarray*}
f(y,y,x,x,x\zd x) &=& x\\
f(y,x,y,x,x\zd x) &=& x\\
f(x,x,x,y,x\zd x) &=& x\\
f(x,x,x,x,y\zd x) &=& x\\
&\vdots&\\
f(x,x,x,x,x\zd y) &=& x.
\end{eqnarray*}

\begin{theorem}[\cite{Berman10:varieties}]\label{the:edge}
For a finite algebra $\zA$ the following conditions are equivalent:\\[1mm]
(a) $\zA$ has few subpowers,\\[1mm]
(b) the variety generated by $\zA$ has an edge term,\\[1mm]
(c) $g_\zA$ is bounded by a polynomial.
\end{theorem}

In this section we will need the property of a finite collection of algebras to 
have few subproducts. More precisely, let $\cK$ be a finite set of 
similar algebras. Let $s_\cK(n)$ be the maximal number of subalgebras of 
a direct product $\zA_1\tm\dots\tm\zA_n$, where $\vc\zA n\in\cK$ are not 
necessarily different. Also, let $g_\cK(n)$ be the least number $k$ such that 
for every subalgebra $\zB$ of $\zA_1\tm\dots\tm\zA_n$ for any 
$\vc\zA n\in\cK$, $\zB$ has a 
generating set containing at most $k$ elements. Set $\cK$ is said to have 
\emph{few subproducts} if $s_\cK(n)$ is bounded by $2^{O(p(n))}$, 
where $p$ is a polynomial in $n$. 
The next statement easily follows from Theorem~\ref{the:edge}.

\begin{corollary}\label{cor:edge}
For a finite set of finite idempotent algebras $\cK$ the following conditions 
are equivalent:\\[1mm]
(a) $\cK$ has few subproducts,\\
(b) the variety generated by $\cK$ has an edge term,\\
(c) $g_\cK$ is bounded by a polynomial.
\end{corollary}

\begin{proof}
Let $\zA=\prod_{\zB\in\cG}\zB$. Since the variety generated by $\zA$ 
equals that generated by $\cK$, it suffices to prove that $\cK$ has few 
subproducts if and only if $\zA$ has few subpowers. Suppose $\zA$ has few subpowers. 
Take $\zB=\zA_1\tm\dots\tm\zA_n$ with $\vc\zA n\in\cK$. Observe that, 
since all the algebras from $\cK$ are idempotent, $\zB$ can be viewed as a 
subalgebra of $\zA^n$, and therefore $s_\cK(n)\le s_\zA(n)$. 
On the other hand $\zA^n$ can be viewed as a product of $|\cK|\cdot n$ 
algebras from $\cK$. Hence, $s_\zA(n)\le s_\cK(|\cK|\cdot n)$, and
therefore is also bounded by $2^{O(p(n))}$, where $p$ is a polynomial.
\end{proof}

\subsection{Semilattice free algebras have few subpowers}%
\label{sec:s-free-few}

Firstly, we observe a simple corollary of 
Proposition~\ref{pro:as-connectivity}(2). By this proposition any two maximal 
elements are connected by a directed asm-path. Since semilattice free 
algebras contain no semilattice edges, every element in such algebras is 
maximal, and every asm-path is a path containing only thin affine 
and majority edges.

\begin{corollary}\label{cor:sf-strongly-connected}
Let $\zA$ be a smooth semilattice free algebra. Then any $a,b\in\zA$ are 
connected with a thin directed path containing only affine and majority 
edges.
\end{corollary}

We now show that every finite collection of semilattice free algebras has 
few subproducts. We use the definition of signature and representation quite similar
to \cite{Berman10:varieties}, except instead of \emph{minority index} we 
use thin affine edges. Let $\rel$ be a subdirect product of $\vc\zA n$, let 
every $\zA_i$ be semilattice free, and let $\Af(\zA_i)$ denote the set of thin 
affine edges of $\zA_i$ (it also contains all pairs of the form $(a,a)$). The 
\emph{signature} is the set 
\begin{eqnarray*}
\Sig(\rel) &=& \{(i,a,b)\mid i\in[n], (a,b)\in\Af(\zA_i), \exists 
\text{ $\ba,\bb\in\rel$}\\
&& \text{with $\ba[i]=a$, $\bb[i]=b$, and $\pr_{[i-1]}\ba=\pr_{[i-1]}\bb$}\}.
\end{eqnarray*}
Note that the pair $a,b$ in this kind of a signature is ordered, because thin
affine edges are directed.
A set of tuples $\rel'\sse\rel$ is a \emph{representation} of $\rel$ if\\[2mm]
(1) for each $(i,a,b)\in\Sig(\rel)$ there are $\ba,\bb\in\rel'$
such that $\ba[i]=a$, $\bb[i]=b$, and $\pr_{[i-1]}\ba=\pr_{[i-1]}\bb$;\\
(2) for each $I\sse[n]$, $|I|\le3$, and every $\ba\in\pr_I\rel$ there is 
$\bb\in\rel'$ such that $\pr_I\bb=\ba$.\\[2mm]
As is easily seen, every representation $\rel'$ of $\rel$ contains a subset 
$\rel''\sse\rel'$ which is also a representation and has size at most
$$
2|\Sig(\rel)|+{n\choose3}\cdot\max\{|\zA_i|\cdot|\zA_j|\cdot|\zA_k|
\mid i,j,k\in[n]\}.
$$

We will need the following lemmas. 

\begin{lemma}[Lemmas 27,~35, \cite{Bulatov20:connectivity}]\label{lem:affine-sl}
Let $\zA_1,\zA_2,\zA_3$ be similar idempotent algebras omitting 
type \one.\\[1mm]
(1) Let $ab$ and $cd$ be thin edges of different types in $\zA_1,\zA_2$, resp. 
Then there is a term operation $r$ with $r(b,a)=b$, $r(c,d)=d$.\\[1mm]
(2) Let $a_1b_1$, $a_2b_2$, and $a_3b_3$ be thin majority edges in 
$\zA_1,\zA_2,\zA_3$, respectively. Then there is a term operation $g'$ 
such that $g'(a_1,b_1,b_1)=b_1$, $g'(b_2,a_2,b_2)=b_2$, 
$g'(b_3,b_3,a_3)=b_3$.
\end{lemma}

\begin{lemma}[Lemma~25, \cite{Bulatov20:maximal}]\label{lem:as-rectangularity} 
Let $\rel$ be a subalgebra of $\zA_1\tm\zA_2$ and let $(a,c)\in\rel$. 
For any $b\in\zA_1$ such that $ab$ is thin edge, 
and any $d\in\zA_2$ such that $cd$ is a thin semilattice or affine edge and 
$(a,d)\in\rel$, it holds $(b,d)\in\rel$.
\end{lemma}

The main result of this section is the following

\begin{theorem}\label{the:semilattice-free-few}
Let $\cK$ be a finite set of finite semilattice free algebras closed under 
subalgebras. Then $\cK$ has few subproducts.
\end{theorem}

\begin{proof}
Let $\rel$ be a subdirect product of $\zA_1\tm\dots\tm\zA_n$, 
$\vc\zA n\in\cK$. We show that any representation of $\rel$ generates $\rel$, 
which will prove that any such $\rel$ has a generating set of size $O(n^3)$, and 
therefore few subproducts. 
Let $\rel'\sse\rel$ be a representation of $\rel$, and $\relo=\Sg{\rel'}$.

Take $\ba\in\rel$; we prove by induction on $k\in[n]$ that 
$\pr_{[k]}\ba\in\pr_{[k]}\relo$. For $k\le3$ it follows from property (2) 
of representations, so assume that $k\ge4$. Suppose that there is 
$\bb\in\relo$ with $\pr_{[k]}\bb=\pr_{[k]}\ba$. Let $\ba[k+1]=a$, 
$\bb[k+1]=b$, and $\zB$ the subalgebra of $\zA_{k+1}$ generated by 
$\{a,b\}$. We will show that $\pr_{[k]}\ba\tm\zB\sse\pr_{[k+1]}\relo$, 
which implies the result. Note that, since 
$(\pr_{[k]}\ba,a),(\pr_{[k]}\ba,b)\in\pr_{[k+1]}\rel$, it also holds that 
$\pr_{[k]}\ba\tm\zB\sse\pr_{[k+1]}\rel$. 
Let $C=\{c\in\zB\mid (\pr_{[k]}\ba,c)\in\pr_{[k+1]}\relo\}$. If $C\ne\zB$ 
then by Corollary~\ref{cor:sf-strongly-connected} there are $c\in C$ and 
$d\in\zB-C$ such that $cd$ is 
a thin majority or affine edge. For the sake of obtaining a contradiction, 
replace $a$ and $b$ with $d$ and $c$, respectively. If $ba$ is an affine 
edge, then as $(\pr_{[k]}\ba,a)\in \pr_{[k+1]}\rel$, the triple
$(k+1,b,a)\in\Sig(\rel)$. Therefore there are $\bc,\bd\in\rel'$ witnessing it,  
and so $(b,a)$ is in the link congruence of $\pr_{[k+1]}\relo$. By 
Corollary~\ref{cor:sf-strongly-connected} there is a path from $\pr_{[k]}\bc$
to $\pr_{[k]}\bb$ consisting of thin affine and majority edges. Since
$(\pr_{[k]}\bb,b)\in\pr_{[k+1]}\relo$ and $ba$ is a thin affine edge, by
Lemma~\ref{lem:as-rectangularity} $(\pr_{[k]}\bb,a)\in\pr_{[k+1]}\relo$,
as well, a contradiction.

Consider now the case when $ba$ is a majority edge. We show that for 
any $J\sse[k]$ there is $\bc\in\relo$ such that $\pr_J\bc=\pr_J\ba$ and 
$\bc[k+1]=a$. For subsets $|J|\le2$ the statement follows from property 
(2) of representations. Take $J\sse[k]$, without loss of generality, 
$J=[\ell]$, and suppose that there are $\ba_1,\ba_2\in\relo$ such that 
$\ba_1[k+1]=\ba_2[k+1]=a$, and 
$\pr_{J-\{\ell-1\}}\ba_1=\pr_{J-\{\ell-1\}}\ba$, 
$\pr_{J-\{\ell\}}\ba_2=\pr_{J-\{\ell\}}\ba$. Let $\zB_1$ be the subalgebra 
of $\zA_{\ell-1}$ generated by $a_1=\ba[\ell-1]$ and 
$b_1=\ba_1[\ell-1]$, and let 
$C_1=\{e\in\zB_1\mid (\pr_{[\ell-2]}\ba,e,\ba[\ell],a)\in
\pr_{[\ell]\cup\{k+1\}}\relo\}$.
As $a_1\not\in C_1$, $C_1\ne\zB_1$, and therefore there are $c_1\in C_1$ and 
$d_1\in\zB_1- C_1$ such that $c_1d_1$ is a thin affine or majority edge. Again,
replace $a_1$ with $d_1$ and $b_1$ with $c_1$. If $b_1a_1$
is an affine edge, by Lemma~\ref{lem:affine-sl}(1)
there is a term operation $r(x,y)$ such that $r(a,b)=a$ and $r(b_1,a_1)=a_1$.
Applying $\bc=r(\ba_1,\bb)$ we obtain a tuple $\bc$ such that $\bc[i]=\ba[i]$
for $i\in[\ell]-\{\ell-1\}$, because $t$ is idempotent, $\bc[\ell-1]=a_1$, and
$\bc[k+1]=a$, a contradiction.

Consider the case when $b_1a_1$ is a majority edge. Let $\zB_2$ be the subalgebra of 
$\zA_{\ell}$ generated by $a_2=\ba[\ell]$ and $b_2=\ba_1[\ell]$, and let
$C_2=\{e\in\zB_2\mid (\pr_{[\ell-1]}\ba,e,a)\in\pr_{[\ell]\cup\{k+1\}}\relo\}$.
As before, we may assume that $b_2a_2$ is a thin majority edge. Then by
Lemma~\ref{lem:affine-sl}(2) there is a term operation $g$ such that 
$g(a_1,a_1,b_1)=a_1$, $g(a_2,b_2,a_2)=a_2$, and $g(b,a,a)=a$.
Therefore for $\bc=g(\bb, \ba_2,\ba_1)$ we have 
$\pr_{[\ell-2]}\bc=\pr_{[\ell-2]}\ba$,
$\bc[\ell-1]=a_1$, $\bc[\ell]=a_2$, and $\bc[k+1]=a$. The result follows.
\end{proof}

\begin{corollary}\label{cor:edge-term}
Let $\cK$ be a finite set of similar semilattice free algebras. Then the variety
generated by $\cK$ has an edge term.
\end{corollary}

\begin{proof}
Let $\gV$ be the variety generated by $\cK$. By 
Theorem~\ref{the:set-preservation} every finite algebra from $\gV$ is 
semilattice free. By Theorem~\ref{the:semilattice-free-few}
it also has few subproducts, and by \cite{Berman10:varieties} $\gV$ has 
an edge term.
\end{proof}

\section{Affine and semilattice: thin edges and undirected connectivity}%
\label{sec:majority-free}

Algebra $\zA$ whose graph does not contain edges of the 
majority type will be called \emph{majority free}. We show that in 
majority free algebras every affine edge contains a pair that is quite similar 
to a thin affine edge, and that any two maximal components of a majority free 
algebra are connected with a thin affine edge.

We start with a simple corollary from Proposition~\ref{pro:as-connectivity}
and a simple observation.

\begin{corollary}\label{cor:max-amax}
Let $\zA$ be a smooth majority free algebra. Then 
$\max(\zA)\sse\amax(\zA)$ and $\zA$ has only one as-component.
\end{corollary}

\begin{proof}
By Proposition~\ref{pro:as-connectivity} every two as-maximal elements $a,b$ 
of $\zA$ are asm-connected. Since $\zA$ 
contains no majority edges this is in fact a directed thin as-path, showing
that $a$ and $b$ are as-connected to each other.

As is easily seen, $\amax(\zA)$ contains a maximal element $a$. Then 
by Proposition~\ref{pro:as-connectivity} $a$ is connected to any 
$b\in\max(\zA)$ with a thin directed asm-path, which in our case is
an as-path. Therefore, $b\in\as(a)=\amax(\zA)$.
\end{proof}

\begin{lemma}\label{lem:double-swap}
Let $\zA$ be a smooth algebra and $C_1,C_2$ its maximal components.
Then there are $a\in C_1$ and $b\in C_2$ such that $(a,b)$ is a maximal
element in the subalgebra of $\zA^2$ generated by $(a,b),(b,a)$.
\end{lemma}

\begin{proof}
Take any $a\in C_1,b\in C_2$ and let $\zC=\Sg{(a,b),(b,a)}$. If they do not 
satisfy the required conditions, let $(c,d)$ be an element maximal in 
$\zC$ and such that $(a,b)\sqq_\zC(c,d)$. Since $\zC$ is symmetric 
with respect to swapping the coordinates, $(d,c)\in\zC$. 
Now, if $\Sgg\zC{(c,d),(d,c)}=\zC$ then $c,d$ satisfy the required 
conditions, and $a,b$ can be replaced with $c,d$. If 
$\zC'=\Sgg\zC{(c,d),(d,c)}\subset\zC$, then we replace $a,b$ with $c,d$ 
and repeat the procedure in $\zC'$. 
\end{proof}

We say that a pair $ab$ from algebra $\zA$ is a \emph{Mal'tsev edge} if 
there exists a term operation of $\zA$ that is Mal'tsev on $\{a,b\}$. Note 
that although Mal'tsev edges have a number of desirable properties, this 
notion is not comparable with the notion of a thin affine edge. We 
use the following result from \cite{Bulatov20:maximal}.

\begin{lemma}[Lemma~12, \cite{Bulatov20:maximal}]%
\label{lem:quotient-edge}
Let $\zA$ be a smooth algebra and $\th\in\Con(\zA)$.
If $ab$ is a thin edge in $\zA$, then $a\fac\th b\fac\th$
is a thin edge in $\zA\fac\th$ of the same type.
\end{lemma}  

The main result of this section is

\begin{theorem}\label{the:affine-thin}
Let $\zA$ be a smooth majority free algebra.
\begin{itemize}
\item[(1)]
Let $C_1,C_2$ be maximal components of $\zA$. There are 
$a\in C_1,b\in C_2$ such that $ab$ is a Mal'tsev edge.
\item[(2)]
Let $a',b'\in\zA$ be such that $a'b'$ is an affine edge and this is 
witnessed by a congruence $\th$ of $\Sgg\zA{a',b'}$. Then there 
are $a\in a'\fac\th, b\in b'\fac\th$ such that $ab$ is a Mal'tsev edge. 
\end{itemize}
Moreover, $ab$ is a Mal'tsev edge for any $a\in C_1,b\in C_2$ 
(or $a\in a'\fac\th,b\in b'\fac\th$ such that $(a,b)$ is as-maximal in the 
subalgebra of $\zA^2$ generated by $(a,b),(b,a)$.
\end{theorem}

\begin{proof}
First, we show that if $(a,b)$ is as-maximal in $\zC=\Sg{(a,b),(b,a)}$
then $ab$ is a Mal'tsev edge. Let $\ba=(a,b)$ and $\bb=(b,a)$. 
It suffices to prove that there exists
a term operation $m$ such that $m(\ba,\ba,\bb)=m(\bb,\ba,\ba)=\bb$.
Consider the relation $\rel$, a subalgebra of $\zC^2$ generated by 
$\{(\ba,\bb),(\ba,\ba),(\bb,\ba)\}$. We need to show that $(\bb,\bb)\in\rel$.
Since $\{\ba\}\tm\zC\sse\rel$, this relation is linked. By 
Corollary~\ref{cor:max-amax} there is only one as-component of $\zC$,
this implies $\amax(\zC)\tm\amax(\zC)\sse\rel$. 
Since $\ba,\bb\in\amax(\zC)$, the result follows.

Now, part (1) follows immediately by Lemma~\ref{lem:double-swap}
and Corollary~\ref{cor:max-amax}. 
For part (2) if $ab$ is affine edge witnessed by a congruence $\th$,
then $\zB\fac\th$, $\zB=\Sg{a,b}$, is a module. Therefore by 
Lemma~\ref{lem:quotient-edge} 
every maximal component of belogns to one of the $\th$-blocks, 
and each $\th$-block contains a maximal component. We can now apply
the argument above to a pair of such maximal components belonging to
$a\fac\th$ and $b\fac\th$.
\end{proof}

Theorem~\ref{the:affine-thin} easily implies that term operations that are
Mal'tsev on a Mal'tsev edges can be uniformized.

\begin{corollary}\label{cor:uniform-maltsev}
Let $\zA$ be a smooth majority free algebra. Then there is a ternary term 
operation $m$ and a set $U$ of a Mal'tsev edges such that
\begin{itemize}
\item[(1)]
$m$ is Mal'tsev on every edge from $U$;
\item[(2)]
for any maximal components $C_1,C_2$ of $\zA$, $C_1\ne C_2$, there are 
$a\in C_1,b\in C_2$ such that $ab\in U$. 
\end{itemize}
\end{corollary}

\begin{proof}
Let $(a_1,b_1)\zd(a_k,b_k)$ be all the pairs of Mal'tsev edges such that 
$a_i\in C,b_i\in C'$, for some maximal components $C,C'$, $C\ne C'$. 
Let $\ba$ be the $2k$-ary tuple $(a_1,b_1,a_2,b_2\zd a_k,b_k)$ and 
$\bb=(b_1,a_1,b_2,a_2\zd b_k,a_k)$. Let also $\zB$ be the subalgebra of
$\zA^{2k}$ generated by $\ba,\bb$. Similar to the proof of 
Theorem~\ref{the:affine-thin} we choose $\bc$ such that 
$\ba\sqq_{\zB}\bc$ and $\bc$ is maximal in the subalgebra generated
by $\bc,\bd$, where $\bd\in\zB$ is the `symmetric' tuple to $\bc$, that is,
if $\bc=(c_1,d_1,c_2,d_2\zd c_k,d_k)$ then 
$\bd=(d_1,c_1,d_2,c_2\zd d_k,c_k)$. Then there exists an operation 
$m$ with $m(\bc,\bc,\bd)=m(\bd,\bc,\bc)=\bd$. This operation is
Mal'tsev on every pair $c_i,d_i$. It remains to observe that 
$(a_i,b_i)\sqq_{\zA^2}(c_i,d_i)$, and therefore $a_i\sqq_\zA c_i$ and 
$b_i\sqq_\zA d_i$.
\end{proof}

If an algebra $\zA$ contains only affine edges (it does not have to be smooth 
in this case) then every element is maximal, and every maximal
component is a singleton. Therefore we obtain the following

\begin{corollary}\label{cor:affine-maltsev}
Let $\zA$ be an algebra that only contains affine edges. Then $\zA$ is Mal'tsev.
\end{corollary}

\section{Semilattice and majority: thin edges and undirected connectivity}%
\label{sec:affine-free}

Algebra $\zA$ whose graph does not contain edges of the 
affine type will be called \emph{affine free}. In this section we prove results
similar to those in Section~\ref{sec:majority-free}. As it will require some 
additional preparations,
in Section~\ref{sec:sm-structure} we remind and improve some results 
about the structure of subdirect products of affine free algebras related 
to maximal components. A pair $ab$ of elements is called a 
\emph{thin undirected majority edge} if it is a majority edge witnessed by the 
equality relation. In other words, if there is a term operation $h$ that 
is majority on $\{a,b\}$. The main result, Theorem~\ref{the:max-majority}
is then proved in Section~\ref{sec:undirected-majority}. 

\subsection{The structure of subdirect products of affine free algebras}%
\label{sec:sm-structure}

We start with reminding several results from \cite{Bulatov20:maximal}.
Recall that an ($n$-ary) relation over a set $A$ is called
\emph{2-decomposable} if, for any tuple $\ba\in A^n$, $\ba\in\rel$ if
and only if, for any $i,j\in[n]$,
$\pr_{ij}\ba\in\pr_{ij}\rel$. The property of 2-decomposability is closely related to
the existence of majority polymorphisms of the relation. In our case
relations in general do not have a majority polymorphism, but they
still have a property close to 2-decomposability. We say that a
relation $\rel$, a subdirect product of $\vc{\zA}n$, is
\emph{quasi-2-decomposable}, if for any elements $\vc an$,
such that $(a_i,a_j)\in\amax(\pr_{ij}\rel)$ for any
$i,j$, there is a tuple $\bb\in\rel$ with $(\bb[i],\bb[j])\in\as{(a_i,a_j)}$
for all $i,j\in[n]$. 

\begin{theorem}[Theorem~30, \cite{Bulatov20:maximal}]%
\label{the:quasi-2-decomp}
Any subdirect product $\rel$ of smooth algebras is quasi-2-decomposable.

Moreover, if $\rel$ is $n$-ary, $X\sse[n]$, tuple $\ba$ is such that
$(\ba[i],\ba[j])\in\amax(\pr_{i,j}\rel)$ for any $i,j$, and $\pr_X\ba
\in\amax(\pr_X\rel)$, there is a tuple $\bb\in\rel$ with
$(\bb[i],\bb[j])\in\as{(\ba[i],\ba[j])}$ for any $i,j\in[n]$, and
$\pr_X\bb=\pr_X\ba$. 
\end{theorem}

\begin{defin}
A relation $\rel\sse \zA_1\tm\ldots\tm \zA_n$ is said to be {\em
almost trivial} if there exists an equivalence relation $\th$ on the set
$[n]$ with classes $I_1\zd I_k$, such that
$$
\rel=\pr_{I_1}\rel\tm\ldots\tm\pr_{I_k}\rel
$$
where $\pr_{I_j}\rel=\{(a_{i_1},\pi_{i_2}(a_{i_1})\zd
\pi_{i_l}(a_{i_1}))\mid a_{i_1}\in \zA_{i_1}\}$, $I_j=\{i_1\zd i_l\}$,
for certain bijective mappings $\pi_{i_2}\colon \zA_{i_1}\to
\zA_{i_2}\zd \pi_{i_l}\colon \zA_{i_1}\to \zA_{i_l}$.
\end{defin}

An algebra $\zA$ is said to be \emph{maximal generated} if it is generated
by one of its maximal components.

\begin{lemma}[Lemma~35, \cite{Bulatov20:maximal}]\label{lem:at-for-simple}
Let $\rel$ be a subdirect product of simple maximal generated 
algebras $\vc\zA n$, say, $\zA_i$ is generated by a maximal component 
$C_i$; and let $\rel\cap(C_1\tm\ldots\tm C_n)\ne\eps$. Then $\rel$ 
is an almost trivial relation.  
\end{lemma}

In the case of affine free algebras this claim can be strengthened by removing
the requirement $\rel\cap(C_1\tm\ldots\tm C_n)\ne\eps$.

\begin{prop}\label{pro:at-for-simple-sm}
Let $\rel$ be a subdirect product of simple maximal generated affine free 
algebras $\vc\zA n$ such that for any $i,j\in[n]$ there is 
$(a,b)\in\pr_{ij}\rel$ such that $a,b$ belong to maximal components 
generating $\zA_i,\zA_j$, respectively. Then $\rel$ is an almost trivial relation.  
\end{prop}

\begin{proof}
It suffices to prove that there are $\vc Cn$, maximal components generating the 
$\zA_i$'s, such that $\rel\cap(C_1\tm\ldots\tm C_n)\ne\eps$. Note first that 
if $\pr_{ij}\rel$ is the graph of a bijective mapping $\pi$, then $\zA_i,\zA_j$
are isomorphic, and $\pi$ is an isomorphism. Therefore if $C_i$ generates
$\zA_i$ then $\pi(C_i)$ is a maximal component that generates $\zA_j$. 
Set $C_j=\pi(C_i)$. For any $\ba\in\rel$ if $\ba[i]\in C_i$, then 
$\ba[j]\in C_j$. Therefore, it suffices to assume that for all $i,j\in[n]$
the relation $\relo=\pr_{ij}\rel$ is not a graph of a bijective mapping.

Since $\zA_i,\zA_j$ are simple, $\relo$ is linked. Suppose that $C_i,C_j$
are maximal components of $\zA_i,\zA_j$, respectively, generating them, 
such that 
$\relo\cap(C_i\tm C_j)\ne\eps$. Then by Proposition~\ref{pro:max-gen}
$C_i\tm C_j\sse\relo$, that is, $\zA_i\tm\zA_j\sse\relo$. For each $\ell\in[n]$
fix an arbitrary maximal component $C_\ell$ generating $\zA_\ell$ and
let $\ba\in\tms\zA n$ be such that $\ba[\ell]\in C_\ell$ for $\ell\in[n]$.
Then for any $i,j\in[n]$ we have $(\ba[i],\ba[j])\in\pr_{ij}\rel$. 
By Theorem~\ref{the:quasi-2-decomp} there is $\ba'\in\rel$ such that
$\ba'[i]\in C_i$ for all $i\in[n]$. This completes the proof.
\end{proof}

We complete this subsection citing the following structural result from 
\cite{Bulatov20:maximal} that will be needed a bit later.

\begin{corollary}[Corollary~38, \cite{Bulatov20:maximal}]\label{cor:max-comp-product}
Let $\rel$ be a subdirect product of smooth algebras
$\zA_1\zd\zA_n$ such that $\pr_{1i}\rel$ is linked for any $i\in\{2\zd n\}$. 
Let also $\ba\in\rel$ be such that $\ba[1]\in\max(\zA_1)$ and 
$\pr_{2\dots n}\ba\in\max(\pr_{2\dots n}\rel)$. Then 
$\as(\ba[1])\tm\as(\pr_{2\dots n}\ba)\sse\rel$.
\end{corollary}

\subsection{Undirected majority edges}\label{sec:undirected-majority}

In tis section we show that unlike in the general case, a (thick) majority
edge always contains a thin undirected  one.

\begin{theorem}\label{the:max-majority}
Let $\zA$ be a smooth affine free algebra.
\begin{itemize}
\item[(1)]
Let $C_1,C_2$ be maximal components of $\zA$. There are 
$a\in C_1,b\in C_2$ such that $ab$ is an undirected thin majority edge.
\item[(2)]
Let $a',b'\in\zA$ be such that $a'b'$ is a majority edge and this is 
witnessed by a congruence $\th$ of $\Sgg\zA{a',b'}$. Then there 
are $a\in a'\fac\th, b\in b'\fac\th$ such that $ab$ is an undirected thin 
majority edge. 
\end{itemize}
Moreover, $ab$ is an undirected thin majority edge for any 
$a\in C_1,b\in C_2$ (or $a\in a'\fac\th,b\in b'\fac\th$ such that $(a,b)$ 
is maximal in the subalgebra of $\zA^2$ generated by $(a,b),(b,a)$.
\end{theorem}

\begin{proof}
Let us first assume that $a,b\in\zA$ are such that $(a,b)$ is maximal in $\zB$, 
the subalgebra of $\zA^2$ generated by $\ba=(a,b)$ and $\bb=(b,a)$. 
Consider the relation $\rel\sse\zB^3$ generated by 
$(\ba,\ba,\bb),(\ba,\bb,\ba),(\bb,\ba,\ba)$. In order to prove that 
$ab$ is a thin majority undirected edge it suffices to show that 
$(\ba,\ba,\ba)\in\rel$. 

Let $C=\se\ba$. Then, as is easily seen, 
$C\tm\zB,\zB\tm C\sse\pr_{ij}\rel$ for $i,j\in[3]$. By 
Theorem~\ref{the:quasi-2-decomp} $\rel\cap C^3\ne\eps$. Set 
$\relo=\Sgg\rel{\rel\cap C^3}$. Then by Lemma~\ref{lem:path-extension} 
$C\sse\pr_i\relo$, and, since $C^2$ is strongly s-connected, 
$C^2\sse\pr_{ij}\relo$. Therefore, we are in the conditions of 
Corollary~\ref{cor:max-comp-product}, implying $C^3\sse\relo$. 
In particular, $(\ba,\ba,\ba)\in\relo\sse\rel$.

Now, part (1) follows immediately by Lemma~\ref{lem:double-swap}. 
For part (2) if $ab$ is a majority edge witnessed by a congruence $\th$,
then $\zD=\Sg{a,b}=a\fac\th\cup b\fac\th$, as $\zA$ is smooth. Therefore 
every maximal component of $\zD$ belogns to either $a\fac\th$ or $b\fac\th$, 
and both $\th$-blocks contain a maximal component. We can now apply
the argument above to a pair of such maximal components.
\end{proof}

Theorem~\ref{the:max-majority} easily implies 

\begin{corollary}\label{cor:sm-connect}
Let $\zA$ be an affine free algebra. Then
any maximal $a,b\in\zA$ are connected by a path
$a=c_1\zd c_k=b$, where each $c_i$ is a maximal element, 
each $c_ic_{i+1}$, $i\in[k-1]$, is a thin semilattice edge or a thin
undirected majority edge, and the path contains at most one majority edge.
\end{corollary}

\section{Semilattice and majority: bounded width}\label{sec:sm}

In this section we give a simpler proof of the Bounded Width Theorem that
a $CSP(\zA)$ has bounded width if and only if $\zA$ omits types \one\ and
\two\ \cite{Barto09:bounded,Barto14:local,Bulatov09:bounded} (or, 
equivalently, if and only if $\zA$ contains no edges of the unary and affine types).

\subsection{Bounded width}

We start with reminding necessary definitions.
Let $\cP=(V,\dl,\cC)$ be a CSP and $W\sse V$. The \emph{restriction} 
of $\cP$ to $W$ is the CSP $\cP_W=(W,\dl\red W,\cC_W)$, where 
$\dl\red W$ is the restriction of $\dl$ on $W$, and for every 
$C=\ang{\bs,\rel}$ the set $\cC_W$ contains the constraint
$C_W=\ang{\bs\cap W,\pr_{\bs\cap W}\rel}$, where $\bs\cap W$ is 
the subtuple of $\bs$ containing all the elements from $W$ in $\bs$, say, 
$\bs\cap W=(\vc ik)$, and $\pr_{\bs\cap W}\rel$ stands for 
$\pr_{\{\vc ik\}}\rel$. A solution of $\cP_W$ is called a 
\emph{partial solution} of $\cP$ on $W$. The set of all partial solutions 
on $W$ is denoted by $\cS_W$. It will be convenient to define problems 
of bounded width as follows. Let $\cF$ be a set of pairs $(W,\vf)$, where
$|W|\le\ell$ and $\vf\in\cS_W$. We say that a tuple $\ba$ over set of 
variables $U$ is \emph{$\cF$-compatible} if for any $W\sse U$, 
$|W|\le k$, it holds $(W,\pr_W\ba)\in\cF$. The set $\cF$ is said to be a 
\emph{$(k,\ell)$-strategy} if the following conditions hold:
\begin{itemize}\itemsep0pt
\item[(S1)]
for every set $W$, $|W|\le \ell$, there is a constraint $\ang{\bs,\rel}\in\cC$
such that $W\sse\bs$;
\item[(S2)]
for any $(W,\vf)\in\cF$ and any $U\sse W$, it holds $(U,\vf\red U)\in\cF$;
\item[(S3)]
for every $\ang{\bs,\rel}\in\cC$, and any $(W,\vf)\in\cF$, $W\sse\bs$,
$|W|\le k$, there is a $\cF$-compatible tuple $\ba\in\rel$ such that
$\pr_W\ba=\vf$,
\end{itemize}
If $\cP$ has a $(k,\ell)$-strategy $\cF$ it is called \emph{$(k,\ell)$-minimal}.

Checking if a CSP is $(k,\ell)$-minimal can be done 
in polynomial time \cite{Dechter03:processing}. For instance the following 
procedure establishes $(k,\ell)$-minimality by finding the largest (with 
respect to inclusion) $(k,\ell)$-strategy.
Start with computing $\cS_W$ for all sets $W\sse V$ with $|W|\le\ell$. 
Then repeat the following steps while this changing the instance: 
(a) Pick $W\sse V$ with $|W|\le k$, $\ba\in\cS_W$, and 
$W\subseteq U\subseteq V$ with $|U|\le\ell$. If there is no $\bb\in\cS_U$ 
such that $\pr_W\bb=\ba$, remove $\ba$ from $\cS_W$. (b) Pick 
$W\sse U\sse V$ with $|W|\le k$, $|U|\le\ell$, and $\bb\in\cS_U$.
If $\pr_W\bb\not\in\cS_W$, remove $\bb$ from $\cS_U$. (c) For every 
$C=\ang{\bs,\rel}$ and $\ba\in\rel$, if there is $W\sse\bs$ with $|W|\le k$
such that $\pr_W\ba\not\in\cS_W$, remove
$\ba$ from $\rel$. As is easily seen, the procedure converges in time 
$O(mn^\ell)$, where $n=|V|$ and $m$ is the total number of tuples 
in all the constraint relations of $\cP$. Also, the resulting collection of 
sets $\cS_W$ is a $(k,\ell)$-strategy by construction, and the resulting 
instance $\cP$ still belongs to $\CSP(\cK)$ provided the original instance 
belongs to $\CSP(\cK)$, and the sets $\cS_W$ are the sets of partial 
solutions of $\cP$. The latter property allows us to assume that a
$(k,\ell)$-strategy can always be chosen to be a collection of sets 
of partial solutions.

Finally, $\CSP(\cK)$ is said to have \emph{width} $(k,\ell)$ if every 
$(k,\ell)$-minimal 
instance of this problem has a solution. A problem is said to have 
\emph{bounded width} if it has width $(k,\ell)$ for some $k,\ell$. Every 
CSP of bounded width has a polynomial time solution algorithm.

\begin{remark}
Note that sometimes $(k,\ell)$-strategy is defined in a different way. The
difference is in the definition of the problem $\cP_W$. The alternative 
definition includes into $\cC_w$ only those constraints $\ang{\bs,\rel}$,
for which $\bs\sse W$. This leads, of course, to a different concept 
of $(k,\ell)$-consistency and width $(k,\ell)$. This alternative concept
of consistency is much weaker, and does not result in an equivalent 
notion of bounded width.
\end{remark}

Problem $\CSP(\cK)$ has bounded width if and only if $\cK$ omits types 
\one\ and \two\
\cite{Larose07:bounded,Bulatov09:bounded,Barto14:local,Barto14:collapse}.
Moreover, a CSP has bounded width if and only if it has width $(2,3)$ 
\cite{Bulatov09:bounded,Barto14:collapse}. There is also an alternative 
characterization in terms of types of edges.

\begin{prop}[Theorem~5, \cite{Bulatov20:connectivity}]\label{pro:sm-omitting2}
For a class $\cK$ of similar idempotent algebras omitting type \one\ 
the following two conditions are equivalent.\\[1mm]
(1) The variety generated by $\cK$ omits types \one\ and \two.\\[1mm]
(2) Algebras from $\cK$ have no edges of the unary and affine types.
\end{prop}

We also give a proof of the following theorem (see 
\cite{Bulatov09:bounded,Barto14:local,Barto14:collapse}).

\begin{theorem}\label{the:main}
Let $\cK$ be a class of similar affine free algebras. Then 
$\CSP(\cK)$ has bounded width. More precisely, every (2,3)-minimal 
instance of $\CSP(\cK)$ has a solution.
\end{theorem}

\subsection{Proof of Theorem~\ref{the:main}}\label{sec:bw}

We start with a simple observation. Clearly, if we replace algebras from 
$\cK$ with their reducts that are also affine free, it suffices to prove 
Theorem~\ref{the:main} for the class of reducts. Therefore by 
Theorem~\ref{the:adding} we may assume that algebras in $\cK$ are 
smooth.

The overall method of proving Theorem~\ref{the:main} is to show that 
starting from a (2,3)-minimal CSP we can construct another CSP,
which is also (2,3)-minimal, but whose domains are smaller. Then we
conclude by induction, in which the base step is a CSP with singleton domains. 
Our reduction restricts one of the domains, say, $\zA$, 
to one of the congruence blocks of a maximal congruence of $\zA$. 

Let $\cP=(V;\dl;\cC)$ be a $(2,3)$-minimal problem instance. To simplify 
notation we use $\zA_v$ rather than $\zA_{\dl(v)}$ for $v\in V$.
We prove by induction on the number of elements in $\zA_v$, $v\in V$, that 
$\cP$ has a solution. If all $\zA_v$, $v\in V$, are 1-element, the result 
holds trivially. 

Suppose that the theorem holds for problem 
instances $\cP'=(V;\dl';\cC')$, where $|\zA_{\dl'(v)}|\le 
|\zA_v|$ for $v\in V$ (here $\zA_{\dl'(v)}$ denotes the set of partial
solutions to $\cP'$ on $\{v\}$), and at least one inequality is strict. 

For some $v\in V$, take a maximal congruence $\th$ of $\zA_v$ (it can be
the equality relation if $\zA_v$ is simple). Note that for any $u\in V-\{v\}$, 
$\cS_{vu}^\th=\{(a\fac\th,b)\mid (a,b)\in\cS_{vu}\}$ is either a linked 
relation, or the graph of a surjective mapping 
$\pi_u\colon\zA_u\to\zA_v\fac\th$. 

Let $W$ denote the set consisting of $v$ and all $u\in V$ such that
$\cS_{vu}^\th$ is the graph of $\pi_u$, and let $\th_u$ denote $\ker\pi_u$,
the congruence of $\zA_u$ which is the kernel of $\pi_u$, for $u\in W$, and 
let $\th_u$ denote the equality relation for $u\in V-W$; also let 
$\th_v=\th$. Since $\cP$ is (2,3)-minimal, for any $u,w\in W$ there is a 
bijective mapping $\pi_{uw}:\zA_u\fac{\th_u}\to\zA_w\fac{\th_w}$ such
that whenever $(a,b)\in\cS_{uw}$, $\pi_{uw}(a\fac{\th_u})=b\fac{\th_w}$.
Take a maximal (as a vertex of $\cG(\zA_v\fac\th)$) $\th$-block $B\sse\zA_v$
and let $\cP'$ be the problem $(V;\dl';\cC')$ given by
$$
\zA_{\dl'(u)}=\left\{\begin{array}{ll}
\pi_{vu}(B) & \hbox{if $u\in W$}, \\
\zA_u & \hbox{otherwise},
\end{array}\right.
$$
and for each $C=\ang{\bs,\rel}\in\cC$ there is $C'=\ang{\bs,\rel'}\in\cC'$  such
that $\ba\in\rel'$ if and only if $\ba\in\rel$ and
$a_u\in\pi_{vu}(B)$ for all $u\in W\cap \bs$.

The result now follows from Lemma~\ref{lem:to-cong-block} by the induction 
hypothesis.

\begin{lemma}\label{lem:to-cong-block}
$\cP'$ is (2,3)-minimal.
\end{lemma}

\begin{proof}
To simplify notation we write $\zA'_u$ rather than $\zA_{\dl'(u)}$.
Let $\cF$ denote the (2,3)-strategy for $\cP$ consisting of the pairs of 
the form $(\{u_1\},a)$, $a\in\cS_{u_1}$, 
$(\{u_1,u_2\},(a,b))$, $(a,b)\in\cS_{u_1u_2}$, 
$(\{u_1,u_2,u_3\},(a,b,c)$, $(a,b,c)\in\cS_{u_1u_2u_3}$.
We show that the set $\cF'$ of the pairs of the form 
$(\{u_1\},\max(\cS'_{u_1}))$, $(\{u_1,u_2\},\max(\cS'_{u_1u_2}))$, 
$(\{u_1,u_2,u_3\},\max(\cS'_{u_1u_2u_3}))$, where 
$\cS'_{u_1}=\zA'_{u_1}$, 
$\cS'_{u_1u_2}=\cS_{u_1u_2}\cap(\zA'_{u_1}\tm\zA'_{u_2})$, and 
$\cS'_{u_1u_2u_3}=\cS_{u_1u_2u_3}\cap
(\zA'_{u_1}\tm\zA'_{u_2}\tm\zA'_{u_3})$,
form a (2,3)-strategy for $\cP'$. Condition (S1) for $\cF'$ immediately 
follows from (S1) for $\cF$. Condition (S2) follows from the 
definition of $\cF'$, (S2) for $\cF$, and Corollary~\ref{cor:path-extension}. 

We will need the following simple statement.

\smallskip

\noindent
{\sc Claim.}
For  $w\in V$ such that $\cS_{vw}^\th$ is not the graph of a mapping,
if $a\in\max(\cS'_w)$ and $a\sqq b$, then $b\in\max(\cS'_w)$.

\smallskip

If $\cS_{vw}^\th$ is not the graph of a mapping, then it is linked. Let $C$
be the maximal component of $\zA_v\fac\th$ containing $B$ and $D$ 
the maximal component of $\zA_w$ containing $a$ and $b$. Since
$\cS_{vw}^\th\cap(C\tm D)\ne\eps$, by Proposition~\ref{pro:max-gen}
$C\tm D\sse \cS_{vw}^\th$. In particular, $b\in\cS'_w$.

\smallskip

To verify (S3) take $u,w\in V$, $(a,b)\in\cS'_{uw}$, and 
$\ang{\bs,\rel}\in\cC$ with $u,w\in\bs$. Without 
loss of generality let $\bs=\{\vc um\}$ and $u=u_1, w=u_2$. 
We need to prove that there exists $\ba\in\rel$ such that 
$\ba[u_1]=a,\ba[u_2]=b$, and for any $s,t\in[m]$ it holds
$(\ba[u_s],\ba[u_t])\in\cS'_{u_su_t}$. We will treat the $\th$-block
$B$ as an element of $\zA_v\fac\th$, and to make notation more uniform
denote it by $e$.

By induction on $i\in[m]$ we prove that there is $\ba_i\in\pr_{[i]}\rel$
such that $\ba[u_1]=a,\ba[u_2]=b$, and for any $s,t\in[i]$ it holds
$(\ba[u_s],\ba[u_t])\in\cS'_{u_su_t}$. The base case, $i=2$, follows
from (S3) for $\cF$. Suppose $\ba_i$ with the required properties 
exists. Consider the relation
\begin{eqnarray*}
\relo(\vc ui,v) &=& \exists u_{i+1}(\pr_{[i+1]}\rel(\vc ui,u_{i+1})\meet
(\cS_{vu_{i+1}}^\th)(v,u_{i+1})\\
& & \quad \meet\bigwedge_{j=1}^i\cS_{u_ju_{i+1}}(u_j,u_{i+1}).
\end{eqnarray*}
We show that $(\ba_i,e)\in\relo$. This implies that there is a value 
$c$ of $u_{i+1}$ that satisfies the required properties. We use 
Theorem~\ref{the:quasi-2-decomp}. As $\ba_i\in\pr_{[i]}\rel$, there is 
$d\in\zA_{u_{i+1}}$ with $(\ba_i,d)\in\pr_{[i+1]}\rel$. Moreover,
by the choice of $\cF$, $(\ba_i[u_j],d)\in\cS_{u_ju_{i+1}}$ for $j\in[i]$. 
There is also a value $b$ of $v$ such that $(\cS_{vu_{i+1}}^\th)(b,d)$.
Hence, $\ba_i\in\pr_{u_1\dots u_i}\relo$. Next, by (S3) for any $j\in[i]$ the 
pair $(e^*,\ba[u_j])\in\cS_{vu_j}$, $e^*\in e$, can be extended to 
$(e^*,\ba[u_j],d)\in\cS_{vu_ju_{i+1}}$. Therefore, there 
is $d\in\zA_{u_{i+1}}$ such that $(\ba[u_j],d)\in\cS_{u_ju_{i+1}}$ 
and $(e,d)\in\cS_{vu_{i+1}}^\th$. Again by (S3) for $\cF$ the pair 
$(\ba[u_j],d)$ can be extended to a tuple $\bb\in\pr_{[i+1]}\rel$,
$\bb[u_j]=\ba_i[u_j]$ and $\bb[u_{i+1}]=d$, and such that for every $s\in[i]$,
it holds $(\bb[u_s],d)\in\cS_{u_su_{i+1}}$. Therefore 
$(\ba_i[u_j],e)\in\pr_{u_jv}\relo$. Theorem~\ref{the:quasi-2-decomp} implies
that $(\ba_i,e')\in\relo$ for some $e'\in\zA_v\fac\th$ with $e\sqq e'$.

There are two possibilities. If, say, $u_{i+1}\in W$ and $u_j\in W$, 
for some $j\in[i]$, then $\pr_{vu_j}\relo$
is the graph of $\pi_{vu_j}$ and $e'=e$. Otherwise 
$\pr_{u_jv}\relo$ is linked, and $C\tm D\sse\pr_{vu_j}\relo$, where
$C$ is the maximal component of $\zA_v\fac\th$ containing $e$ and
$D$ is that of $\zA_{u_j}$ containing $\ba[u_j]$. Since 
$\ba_i\in\max(\pr_{[i]}\relo)$, by Corollary~\ref{cor:max-comp-product}
we conclude that $(\ba_i,e)\in\relo$.

The result follows.
\end{proof}

\bibliographystyle{plain}

\end{document}